\documentclass[12pt]{article}
\usepackage{graphicx}
\usepackage{natbib} 
\usepackage{url} 

\newcommand{\blind}{0}

\usepackage{authblk}
\usepackage{amsmath}
\usepackage{amssymb}
\usepackage{tikz}
\usetikzlibrary{arrows}
\usepackage{float} 
\usepackage{subcaption}
\usepackage[toc,page]{appendix}

\usepackage{booktabs}
\usepackage{algorithm}
\usepackage{algpseudocode}
\usepackage{bbm}
\usepackage{adjustbox}
\newcommand{\bi}{\begin{itemize}}
\newcommand{\ei}{\end{itemize}}

\usepackage[shortlabels]{enumitem}
\usepackage{enumitem}

\newcommand{\E}{\mathbb{E}}

\newcommand{\T}{\intercal}

\newcommand{\bX}{\mathbf{X}}
\newcommand{\bx}{\mathbf{x}}
\newcommand{\bS}{\mathbf{S}}
\newcommand{\ba}{\mathbf{a}}
\newcommand{\bA}{\mathbf{A}}
\newcommand{\btau}{\pmb{\tau}}
\newcommand{\bY}{\mathbf{Y}}

\newcommand{\bH}{\mathbf{H}}

\newcommand{\bV}{\mathbf{V}}
\newcommand{\btheta}{\pmb{\theta}}

\newcommand{\bbeta}{\pmb{\beta}}
\newcommand{\bZ}{\pmb{Z}}

\newcommand{\bmu}{\pmb{\mu}}

\newcommand{\bQ}{\mathbf{Q}}
\newcommand{\bU}{\mathbf{U}}
\newcommand{\bI}{\mathbf{I}}

\newcommand{\bD}{\mathbf{D}}
\newcommand{\bSigma}{\pmb{\Sigma}}
\newcommand{\bOmega}{\pmb{\Omega}}
\newcommand{\bLambda}{\pmb{\Lambda}}
\newcommand{\bgamma}{\pmb{\gamma}}
\newcommand{\balpha}{\pmb{\alpha}}
\newcommand{\bUpsilon}{\pmb{\Upsilon}}

\newcommand{\bGamma}{\pmb{\Gamma}}
\newcommand{\bC}{\mathbf{C}}

\newcommand{\bCalO}{\pmb{\mathcal{O}}}

\newcommand{\rtn}{\sqrt{n}}
\newcommand{\pn}{\mathbb{P}_n}
\newcommand{\bPsi}{\pmb{\Psi}}
\newcommand{\bpsi}{\pmb{\psi}}
\newcommand{\tr}{\operatorname{tr}}

\newtheorem{thm}{\bf Theorem}[section]
\newtheorem{prop}{\bf Proposition}[section]
\newtheorem{lemma}{\bf Lemma}[section]

\newenvironment{proof}{\paragraph{Proof:}}{\hfill$\square$}
\usepackage{caption}
\makeatletter
\g@addto@macro\normalsize{%
  \setlength\abovedisplayskip{3pt}
  \setlength\belowdisplayskip{3pt}
  \setlength\abovedisplayshortskip{3pt}
  \setlength\belowdisplayshortskip{3pt}
}
\makeatother

\begin{document}

\def\spacingset#1{\renewcommand{\baselinestretch}%
{#1}\small\normalsize} \spacingset{1}


\if0\blind
{
  \title{\bf Optimal Data Integration and Adaptive Sampling for Efficient Treatment\\ Effect Estimation}
  \small
   \author[1]{Yen-Chun Liu}
    \author[1]{Alexander Volfovsky}
     \author[2]{\authorcr German Schnaidt}
     \author[2]{Cristobal Garib}
     \author[1]{Eric Laber}
    \affil[1]{Department of Statistical Science, Duke University}
     \affil[2]{Amazon Ads ECON}
  \maketitle
} \fi

\if1\blind
{
  \bigskip
  \bigskip
  \bigskip
  \begin{center}
    {\LARGE\bf }
\end{center}

  \medskip
} \fi

\bigskip
\begin{abstract}
This study addresses the challenge of estimating average treatment effects (ATEs) for advertising campaigns in online marketplaces where complete randomized experimentation is infeasible. We propose two key innovations: (1) a shrinkage estimator that optimally combines observational and experimental data without assuming smooth treatment effects across campaigns, and (2) a Bayesian adaptive experimental design framework that efficiently selects campaigns for randomized evaluation that minimizes cumulative risk. Our shrinkage estimator achieves lower risk compared to existing methods by balancing bias-variance tradeoffs, while our adaptive design significantly reduces the costs of campaign randomization. We establish theoretical guarantees including asymptotic normality and regret bounds. In an application to Amazon Ads data analyzing 2,583 campaigns, our approach achieves equivalent estimation precision while requiring only half of the randomized experiments needed by random sampling, the standard method widely used in practice today. The proposed method serves as a practical solution for marketplace platforms to efficiently measure advertising effectiveness while managing experimentation costs.
\end{abstract}

\noindent%
\vfill

\newpage
\spacingset{1.7} 
\section{Introduction}\label{sec:intro}

Online marketplaces such as Amazon, AliExpress, eBay, and Etsy allow sellers to compete for advertising space to increase product exposure and sales. With online advertising comprising a substantial portion of sellers' marketing budgets, understanding campaign effectiveness is critical. While the 
gold standard for estimating the ATE is a randomized experiment, running such experiments
for each campaign is expensive, inefficient, and may harm the customer experience 
\citep[][]{lewis2015unfavorable, gui2020combining}.  
Conversely, large observational data are collected
as a matter of course, but they may be subject to unmeasured confounding and lead leading to biased estimation.  
Thus, online marketplace advertising services
are looking for ways of combining large observational data with experimental
data collected through a small number of judiciously chosen randomized studies.  

We address both the problem of \textbf{how to fuse observational and experimental data} \textit{and} \textbf{how to select which future campaigns to evaluate via randomized experiments}.
For the first task, we propose a shrinkage estimator of the ATE across all campaigns 
that is valid even
when randomized study data are available for only
a subset of the campaigns. We combine a regression-based de-biasing approach and shrinkage estimators to minimize risk in estimating ATEs, and derive an expression for the risk under weighted squared error loss. For the second challenge, we propose a Bayesian adaptive experimental design framework that sequentially selects campaigns for randomization. Instead of relying on arbitrary rules or simple random sampling, our approach uses Thompson sampling to balance exploration and exploitation while systematically minimizing cumulative risk. 

Efficiently combining observational and RCT data is crucial for studying causal effects 'in the wild' \citep{colnet2020causal, degtiar2021review}. When both data sources are available, RCT data can be used for de-biasing observational estimates \citep{kallus2018removing, yang2020improved}. However, this involves a bias-variance trade-off: observational data have larger sample sizes but are potentially biased, while RCT data provide unbiased estimates with higher variance. Shrinkage estimators optimize this trade-off by combining these sources to minimize overall risk, such as the weighted squared error of estimated ATEs \citep{chen2015data, fourdrinier2018shrinkage, rosenman2020combining}. Additionally, recent work in causal inference has explored adaptive design under limited experimental resources, such as sequential optimization for intervention effects \citep{dawson2008sequential,zhang2023active} and active learning for causal discovery \citep{toth2022active}. While existing designs typically focus on a single data source, our work is the first to develop adaptive experimental design methods that leverage both observational and RCT data to minimize estimation risk.

Our contributions are threefold. First, we address the practical setting of multiple concurrent interventions with spillover effects, which is common in online marketplaces where advertisers compete for space and customers see multiple advertisements. Second, we develop a unified framework that jointly considers optimal data fusion and adaptive experimental design. Finally, we establish theoretical guarantees for both components of our approach.  In Section~\ref{sec:setup}, we set notation and introduce our de-biased regression model. Section~\ref{sec:method} presents our proposed methods. Section~\ref{sec:simu} shows simulation results and Section~\ref{sec:da} illustrates the Amazon Ads application. Concluding remarks are given in Section~\ref{sec:discussion}.

\section{Background}\label{sec:setup}

\subsection{Notation and setup}\label{sec:notation}
We assume that there are $J$ binary interventions to be evaluated concurrently in each
of $m=1,\ldots, M$ rounds.  In an online marketing application, the interventions might
be the advertising campaigns for $J$ brands while the rounds are days or weeks.  
During each round, 
interventions can be independently turned `on' or `off' for each unit (customer, patient,
etc.). 
In round $m$, one source of information about the effectiveness of interventions is 
a large observational study 
\begin{equation*}
    \mathcal{O}_{m} = \left\lbrace 
    (\bX_{i,m}, \bA_{i,m}, Y_{i, m})
    \right\rbrace_{i=1}^{N_m},
\end{equation*}
which comprises $N_m$ independent copies of 
$(\bX, \bA, Y)$, where: $\bX\in\mathcal{X}\subseteq \mathbb{R}^{p_x}$ is contextual
information, e.g., customer and campaign information; 
$\bA=(A^1,...,A^J) \in \mathcal{A} = \left\lbrace 0, 1\right\rbrace^J$ is intervention status 
with $A^j$ equal to one 
if intervention
$j$ is on and zero otherwise; and $Y\in\mathcal{Y}$ is an outcome of interest. 
To account for unobserved confounding, denoted as $\bU \in\mathcal{U}$, we 
assume intervention $\bA$ is determined by a function $f(\bX,\bU)$, where $f:\mathcal{X}\times \mathcal{U}\rightarrow \mathcal{A}$ is
an unknown mapping that needs to be estimated. Note that here we assume that the interaction between treatment assignments is fully determined by contextual information and unobserved confounders.

In addition to the observational study, we assume that in each round $m$, we can select a subset of the interventions,  
$\mathbf{S}_m \subseteq \left\lbrace 1,\ldots, J\right\rbrace$, to evaluate via auxiliary 
randomized experiments. Suppose randomized data in round $m$ are of the form 
\begin{equation*}
    \mathcal{R}_m(\bS_m) = \left\lbrace
    (\bX_{\ell,m}, W_{\ell,m}, \bA_{\ell,m}, Y_{\ell,m})
    \right\rbrace_{\ell=1}^{L_m},
\end{equation*}
which comprises $L_m$ independent copies of $(\bX, W, \bA, Y)$, where $W_{\ell,m}\in\{1,...,J\}$ is a single intervention uniformly drawn from $\bS_m$ to be randomized in copy $\ell$. The context 
$\bX$, the outcome $Y$, and the unobserved
confounder $\bU$
are distributed 
as in the observational data. The $j-$th intervention $A^j$ is constructed from $W$
as follows: 
\begin{equation*}
    \left\{
    \begin{array}{lc}
      A^j \sim \operatorname{Bernoulli}(0.5)   & \text{ if } W=j, \\
       A^j = f^j(\bX,\bU) & \text{otherwise},
    \end{array}\right.
\end{equation*}
where $f^j(\bX,\bU)$ is the $j-$th component of $f(\bX,\bU).$ Thus, for each copy $\ell$ in each round $m$, we uniformly select one intervention from $\bS_m$ to evaluate in randomized experiments,  while generating all other interventions as the same process as in the observational study.

Before characterizing intervention effects, we first introduce the concept of a behavior policy.
Let $\mathcal{P}(\mathcal{A})$ denote the space of probability distributions over $\mathcal{A}$, and define a {\em behavior 
policy} as a map $\mu:\mathcal{X}\times \mathcal{U}\rightarrow \mathcal{P}(\mathcal{A})$. In other words, a behavior policy is a probability distribution over the space of intervention status conditional on inputs. Given context $\bx$ and confounder $\mathbf{u}$, under 
policy $\mu$, an intervention $\ba\in\mathcal{A}$ is selected with probability $\mu(\bx, \mathbf{u})(\ba)$. Let $Y^*(\ba)$ be the potential outcome under
$\ba$ and let $\mathcal{Y}^* = \left\lbrace Y^*(\ba)\,:\, \ba \in\mathcal{A}\right\rbrace$ be
the collection of all potential outcomes. For a policy $\mu$, we define a collection of mutually 
independent random variables $\left\lbrace \bA_\mu(\bx, \mathbf{u})\,:\, (\bx,\mathbf{u})\in\mathcal{X}\times\mathcal{U}\right\rbrace$, independent of $\left\lbrace \bX, \bU, \bA,  
\mathcal{Y}^*\right\rbrace$, such that 
$P\left\lbrace \bA_\mu(\bx, \mathbf{u}) = \ba \right\rbrace = \mu(\bx, \mathbf{u})(\ba)$
for all $(\bx, \mathbf{u}, \ba)$ .
The potential
outcome under a behavior policy $\mu$ is thus defined as
\begin{equation*}
    Y^*(\mu) = \sum_{\ba\in\mathcal{A}}Y^*(\ba)1_{\ba = \bA_\mu(\bX, \bU)}.
\end{equation*}
For each intervention $j$, 
define 
$\mu_{0}^j:\mathcal{X}\times \mathcal{U}\rightarrow \mathcal{P}(\mathcal{A})$ as the behavior policy that turns intervention $j$ off while ensuring other interventions are generated as the same process as in the observational study. That is, we define the $k-$th component of the mapping as
$\mu_{0,k}^j(\bx, \mathbf{u}) = f^k(\bx, \mathbf{u})1_{k \ne j}$, $k=1,\ldots,J$. Similarly, 
define $\mu_1^j$ to be the behavior policy such that 
$\mu_{1,k}^j(\bx, \mathbf{u}) = f^k(\bx, \mathbf{u})1_{k\ne j} + 1_{k = j}$, so that the $j$-th intervention is on while the others are generated as in the observational study. The
average treatment effect in the wild (ATE-ITW) of intervention $j$ is then defined as 
\begin{equation*}
    \tau^{j*} \triangleq 
    \mathbb{E}
   \left\{ Y^*(\mu_1^j)\right\}
    - \mathbb{E}\left\{
    Y^*(\mu_0^j)\right\}.
\end{equation*}


In the terminology of dynamic treatment
regimes, $\tau^{j*}$ might be termed the {\em blip-to-reference} with business-as-usual
as the reference \citep[][]{moodie2007demystifying}.    
Let $\btau^* \triangleq (\tau^{1*},\ldots, \tau^{J*})$. Given any
estimator $\widehat{\btau}$ of $\btau^*$, we consider weighted squared error loss
\begin{equation*}
    \mathfrak{L}_{\bD}(\widehat{\btau}, \btau^*) \triangleq
    \left(
    \widehat{\btau} - \btau^*
    \right)^{\T}\bD\left(
    \widehat{\btau} - \btau^*
    \right),
\end{equation*}
 where $\bD \in \mathbb{R}^{J\times J}$ is a symmetric positive
 definite weight matrix. We aim to 1) construct an estimator 
$\widehat{\btau}$
of $\btau^*$ minimizing expected loss (i.e., risk), and 2) efficiently select the interventions  
to evaluate via randomized study in each round so 
as to minimize cumulative risk.

\subsection{Estimation of intervention effects}\label{sec:exist_estimator}
Recall that $W_{\ell,\nu}\in\{1,\ldots,J\}$ denotes the intervention that is randomized in copy $\ell$ in round $m$. Suppose for all $j \in \bigcup_{\nu = 1}^m \bS_\nu$, there exist some $1\leq \nu\leq m$ and $1\leq\ell\leq L_{\nu}$ such that $W_{\ell,\nu}=j$. That is, assume that each intervention available for randomization by round $m$ has been selected at least once. We define the RCT estimator of $\tau^{j*}$ as
 
\begin{align*}
    \widehat{\tau}_{\mathcal{R},m}^j& = 
 \frac{   \sum_{\nu=1}^m
    \sum_{\ell =1}^{L_{\nu}} 
    \mathbf{1}_{W_{\ell,\nu}= j}
    A_{\ell,\nu}^j Y_{\ell,\nu}}{\sum_{\nu=1}^m
    \sum_{\ell =1}^{L_{\nu}} 
    \mathbf{1}_{W_{\ell,\nu}= j}A_{\ell,\nu}^j} - 
    \frac{
    \sum_{\nu=1}^m
    \sum_{\ell=1}^{L_{\nu}} \mathbf{1}_{W_{\ell,\nu}= j}(1-A_{\ell,\nu}^j) Y_{\ell, \nu}
    }{  \sum_{\nu=1}^m\sum_{\ell =1}^{L_{\nu}} 
    \mathbf{1}_{W_{\ell,\nu}= j}(1-A_{\ell,\nu}^j)}.
\end{align*}
 
While the RCT estimator is unbiased, using only randomized data is suboptimal for two reasons. First, it overlooks the information in the observational study, which typically has a sample size orders of magnitude larger. Second, in online marketing applications, only a small fraction of the $J$ total interventions can be evaluated via RCT study. To overcome these limitations, we leverage observational data by constructing a doubly robust estimator (DR; \citet{funk2011doubly}) of each $\tau^{j*}$. Define $\overline{N}_m:=\sum_{\nu=1}^mN_\nu$ as the cumulative observational sample size. The DR estimator of $\tau^{j*}$ is defined as 
\begin{align*}
        \widehat{\tau}_{\mathcal{O},m}^j 
      &  = \frac{1}{\overline{N}_m}\sum_{\nu=1}^m
    \sum_{i=1}^{N_{\nu}}  \widehat{m}^j_1(\bX_{i,\nu}) + \frac{A^j_{i,\nu} \{Y_{i,\nu} - \widehat{m}^j_1(\bX_{i,\nu})\}}{\widehat{P}_m(A_{i,\nu}^j|\bX_{i,\nu})}\\
    &-\frac{1}{\overline{N}_m}\sum_{\nu=1}^m
    \sum_{i=1}^{N_{\nu}}  \widehat{m}^j_0(\bX_{i,\nu}) + \frac{(1-A^j_{i,\nu}) \{Y_{i,\nu} - \widehat{m}^j_0(\bX_{i,\nu})\}}{1-\widehat{P}_m(A_{i,\nu}^j|\bX_{i,\nu})},
\end{align*}
where $\widehat{P}_m(A_{i,\nu}^j|\bX_{i,\nu})$ is the estimated propensity score and $\widehat{m}^j_1(\bX_{i,\nu}) $ and $\widehat{m}^j_0(\bX_{i,\nu})$ are regression outcome estimates of expected potential outcomes in the treatment and control groups, respectively.

Due to  unmeasured confounding,  
$\widehat{\tau}_{\mathcal{O},m}^j$  might not be consistent for $\tau^{j*}$ \citep{vermeulen2015bias}. To de-bias 
these estimators, we assume each intervention is associated with a vector of attributes $\mathbf{V}^j \in\mathcal{V}$, 
$j=1,\ldots, J$; 
e.g., in the context of online marketing, these attributes might characterize
the nature of an advertising campaign including channel(s), brand reputation,
market share, and so on. We assume that 
$$
\mathbb{E}\left(\widehat{\tau}_{\mathcal{O}, m}^j \mid \mathbf{V}^j=\mathbf{v}^j\right)=\tau^{j *}+\psi\left(\mathbf{v}^j\right)^{\top} \boldsymbol{\theta}^*+\rho_m^j
$$

where $\psi: \mathcal{V} \rightarrow \mathbb{R}^{p_v}$ is a user-specified feature mapping, $\boldsymbol{\theta}^* \in \boldsymbol{\Theta} \subseteq \mathbb{R}^{p_v}$ is an unknown parameter, and $\rho_m^j$ is a remainder satisfying $\sup _j\left|\rho_m^j\right|=o\left\{\left(\sum_{\nu=1}^m N_\nu\right)^{-1 / 2}\right\}$.  
The proposed bias model has two important features.  
First, the remainder term $\rho_m^j$ relaxes the assumption of unbiasedness to consistency. Second, we assume smoothness only in the bias structure across $\mathbf{v}$, {\em not in the treatment effects} which reflects 
the practical reality that similar brands can have dramatically different campaign effectiveness \citep{kaptein2012heterogeneity}.

Let $\overline{\bS}_m:=\cup_{\nu=1}^m\bS_m$ and  suppose $|\overline{\bS}_m|\geq p_v$. Define the least squares estimator of $\boldsymbol{\theta}^*$ as
$$
\widehat{\boldsymbol{\theta}}_m\left(\overline{\mathbf{S}}_m\right) = \arg \min _{\boldsymbol{\theta}} \sum_{j \in \overline{S}_m}\left\{\widehat{\tau}_{\mathcal{O}, m}^j-\widehat{\tau}_{\mathcal{R}, m}^j-\psi\left(\mathbf{V}^j\right)^{\top} \boldsymbol{\theta}\right\}^2.
$$

Denote $\boldsymbol{\Psi}$ as the $p_v \times J$ matrix whose $j$-th column equals to $\psi\left(\mathbf{V}^j\right)$. Under relatively mild conditions given in the Appendix, it follows that $\widehat{\boldsymbol{\tau}}_{\mathcal{O}, m}-\boldsymbol{\Psi} \widehat{\boldsymbol{\theta}}_m-\boldsymbol{\tau}^*$ is asymptotically normal with mean zero as $m$ grows large. While $\widehat{\boldsymbol{\tau}}_{\mathcal{O}, m}-\boldsymbol{\Psi} \widehat{\boldsymbol{\theta}}_m$ consistently estimates $\boldsymbol{\tau}^*$, it has relatively higher variance compared to $\widehat{\boldsymbol{\tau}}_{\mathcal{O}, m}$, due to $\widehat{\boldsymbol{\theta}}_m$ being estimated using the 
(much smaller) randomized data. In the next section, we propose a shrinkage estimator that optimizes the bias-variance trade-off by minimizing weighted risk.
\section{Proposed methods}\label{sec:method}

\subsection{Optimal shrinkage estimator $\widehat{\btau}^\lambda_m$}\label{sec:opt_est}

We consider estimators of $\btau^*$ of the form
\begin{equation}\label{eq:shrink_class}
    \widetilde{\btau}_{m}^{\lambda}(\overline{\bS}_m) = 
    \widehat{\btau}_{\mathcal{O},m} - (1-\lambda) \pmb{\Psi}\widehat{\btheta}_{m}(\overline{\bS}_m), 
\end{equation}
where the dependence on $\overline{\bS}_m$, historical interventions evaluated in randomized studies, has been made explicit for convenience 
when we consider optimal design. When $\lambda=1$, $\widetilde{\btau}_{m}^{\lambda=0}(\overline{\bS}_m)$ reduces to the DR estimator $\widehat{\btau}_{\mathcal{O},m}$, which is efficient but can be seriously biased. Conversely, when $\lambda=0$, $\widetilde{\btau}_{m}^{\lambda=1}(\overline{\bS}_m)$ is equivalent to the fully-de-biased estimator, $\widehat{\btau}_{\mathcal{O},m}-\pmb{\Psi}\widehat{\btheta}_{m}(\overline{\bS}_m)$, which is consistent for the oracle treatment effects $\btau^*$ but prone to high variance. 
The objective is to tune $\lambda\in[0,1]$ to minimize the estimated risk when
$\overline{\bS}_m$ is fixed.

To derive an optimal shrinkage estimator $\widetilde{\btau}^\lambda_m$, we will 
first make use of the following 
Lemma derived from similar results in \citet{strawderman2003minimax} and
\citet{rosenman2020combining}  \citep[see also][]{fourdrinier2018shrinkage} for risk estimation.

\begin{lemma}
\label{lemma:risk}
Suppose that $\bZ \sim N_p(\pmb{\mu}, \pmb{\Sigma})$ and $\pmb{Y}$ is a 
random vector in $\mathbb{R}^q$. Define 
the weighted squared error loss 
\begin{equation*}
    \mathcal{L}_{\bD}(\pmb{\nu}, \bmu) = (\pmb{\nu} - \bmu)^{\T}\bD(\pmb{\nu} - \bmu),
\end{equation*}
where $\pmb{D}$ is a fixed positive definite matrix.  
Let $g: \mathbb{R}^p\times \mathbb{R}^q\rightarrow \mathbb{R}^p$ 
be differentiable and satisfy $\mathbb{E}||g(\bZ, \bY)||^2 < \infty$. 
Then the risk of $\kappa(\bZ, \bY):= \bZ + \pmb{\Sigma} g(\bZ, \bY)$ is 
given by 
\begin{align*}
    &\mathfrak{R}(\bD, \bSigma, g)
   =
    \frac{1}{p}\mathbb{E}\Bigg(
    \sum_{j=1}^p \lambda_j \left[\left\lbrace
    \bOmega\bSigma g(\bZ, \bY)
    \right\rbrace_j^2 
    + 2\frac{\partial \left\lbrace \bOmega\bSigma g(\bZ, \bY) \right\rbrace_j}{\partial (\bOmega \bZ)_j^2}
    \right]\Bigg) +  
    \frac{1}{p}\mathrm{tr}(\bLambda),
\end{align*}
where $\bOmega$, $\bLambda$, and $\lambda_j$ are given in the Appendix.
\end{lemma}
\noindent 
A proof is provided in the Appendix. Lemma~\ref{lemma:risk} provides a closed form of risk given a weight matrix $\bD$, the covariance of an estimator $\bZ$, and a shrinkage function $g(\cdot,\cdot)$. While the expression appears complex, the shrinkage estimator \eqref{eq:shrink_class} has the required form to apply Lemma \ref{lemma:risk}:
\begin{equation}\label{eq:shrinkage}
    \widetilde{\btau}_m^{\lambda}(\overline{\bS}_m) = 
\underbrace{\widehat{\btau}_{\mathcal{O},m} - 
\bPsi \widehat{\btheta}_m(\overline{\bS}_m)}_{\bZ} \nonumber+ \bSigma\underbrace{\left(\lambda\bSigma^{-1}\left[\widehat{\btau}_{\mathcal{O}, m} - \left\{\widehat{\btau}_{\mathcal{O},m} - 
\bPsi \widehat{\btheta}_m(\overline{\bS}_m)\right\}\right]\right)}_{g^{\lambda}(\bZ, \bY)}.
\end{equation} Here, $\bZ$ corresponds to the fully-debiased estimator $\widehat{\btau}_{\mathcal{O},m} - 
\bPsi \widehat{\btheta}_m(\overline{\bS}_m)$, $\bY$ is the DR estimator $\widehat{\btau}_{\mathcal{O},m}$, and $g^\lambda(\bZ,Y)=\lambda\bSigma^{-1}(\bY-\bZ)$ is the estimated bias vector shrunk by $\lambda$ and scaled by $\bSigma^{-1}$. While our proposed shrinkage estimator resembles the approach in \citet{rosenman2020combining}, a key difference is our use of the debiased estimator instead of direct RCT estimates as proxies for true treatment effects. This makes our method more general since it does not require RCT estimates for all interventions, which is common in online marketing experiments where complete randomized data is rarely available. Instead, we only require at least $p_v$ RCT estimates. This allows us to fuse observational and RCT data even when randomized data is incomplete. In the special case where only a single intervention is considered, our estimator in \eqref{eq:shrinkage} reduces to that in \citet{rosenman2020combining}.

We first establish the asymptotic normality of the fully de-biased estimator before proving the asymptotic normality of our proposed shrinkage estimator. Let $\bGamma$ be the
asymptotic variance of $\widehat{\btau}_{\mathcal{O},m}$ and let $\bUpsilon$ be the asymptotic variance of $\widehat{\btau}_{\mathcal{R},m}$. Define $\overline{L}_m=\sum_{\nu=1}^mL_{\nu}$ as the cumulative RCT sample size. The asymptotic behavior of the fully-debiased estimator $\bZ=\widehat{\btau}_{\mathcal{O},m} - 
\bPsi \widehat{\btheta}_m(\overline{\bS}_m)$ is summarized in Lemma~\ref{lemma: asyZ}.

\begin{lemma}\label{lemma: asyZ}
    Suppose that $\bZ$ is defined as in~(\ref{eq:shrinkage}) and assume that $|\overline{\bS}_m|\geq p_v$. Suppose $\lim_{m\to\infty}\overline{N}_m/(\overline{N}_m+\overline{L}_m)=\rho\in(0,1)$. Under regularity conditions (B1)-(B4) and (C1)-(C3), it can be shown that $\sqrt{\overline{N}_m^{-1}+\overline{L}_m^{-1}}^{-1}\bZ\overset{d}{\to} N_p(\btau^*,\bSigma)$, where
    \begin{equation}\label{eq:Sigma}
        \bSigma = (1-\rho)(\bI-\bH)\bGamma(\bI-\bH)^\T + \rho \bH \bUpsilon\bH^\T
    \end{equation}

with $\bH$ given in the Appendix. 
\end{lemma}
Plugging $\bZ$, $g$, and $\bSigma$ specified in \eqref{eq:shrinkage} and \eqref{eq:Sigma} into ${\mathfrak{R}}(\bD, {\bSigma}, g^\lambda)$ yields the risk of $\widetilde{\btau}_m^\lambda(\overline{\bS}_m)$. Since $\bSigma$ and $\bGamma$ are unknown in practice, we replace them with their respective plug-in estimators, $\widehat{\bSigma}_m$ and $\widehat{\bGamma}_m$. By substituting these estimators and approximating the expectation with observed realizations, we derive an empirical unbiased risk estimator (eURE) of the shrinkage estimator $\widetilde{\btau}_m^{\lambda}(\overline{\bS}_m)$:
\begin{align*}
      \widehat{\mathfrak{R}}_m(\bD, \widehat{\bSigma}_m, g^\lambda)&=\tr(\bD\widehat{\bSigma}_m) +\lambda^2 \E \left\{\left(\bPsi\widehat{\btheta}_m\right)^\T\bD\bPsi\widehat{\btheta}_m\right\}-2\lambda\tr\left\{-\bD\widehat{\bGamma}_m\left(\mathbf{I}-\bH_m\right)^\T +\bD\widehat{\bSigma}_m\right\},  
\end{align*}
where a proof and closed forms for $\widehat{\bSigma}_m$ and $\bH_m$ are given in the Appendix. The optimal shrinkage parameter is thus defined as \[\widehat{\lambda}_m \in \arg\min_{\lambda}\widehat{\mathfrak{R}}_m\left\lbrace \bD, 
\widehat{\bSigma}_{m}, g^{\lambda}\right\rbrace.\] The following theorem summarizes the closed form of the optimal shrinkage factor and the corresponding eURE. The proof is provided in the Appendix.


\begin{thm}\label{thm:opt_shrink}
  Under weighted squared error loss with weights $\bD$, the optimal shrinkage factor is given by
    \begin{equation*}
        \widehat{\lambda}_{m} = \frac{\tr\left\{-\bD\widehat{\bGamma}_{m}\left(\mathbf{I}-\bH_m\right)^\T + \bD\widehat{\bSigma}_m\right\}}{(\Psi\widehat{\btheta}_m)^\T\bD\Psi\widehat{\btheta}_m},
    \end{equation*}
    and the corresponding eURE is
    \begin{align*}
        \operatorname{eURE}(\widehat{\btau}^\lambda_m,\btau^*) &= \tr(\bD\widehat{\bSigma}_m) - \frac{\tr\left\{-\bD\widehat{\bGamma}_{m}\left(\mathbf{I}-\bH_m\right)^\T + \bD\widehat{\bSigma}_m\right\}^2}{(\Psi\widehat{\btheta}_m)^\T\bD\Psi\widehat{\btheta}_m}.
    \end{align*}
    Closed forms of $\widehat{\bSigma}_m$ and $\bH_m$ are provided in the Appendix.
\end{thm}
Theorem~\ref{thm:opt_shrink} extends the estimator of \citep{rosenman2020combining} by accounting for the correlation between the fully-de-biased estimator $\widehat{\btau}_{\mathcal{O},m} - 
\bPsi \widehat{\btheta}_m(\overline{\bS}_m)$ and the DR estimator $\widehat{\btau}_{\mathcal{O},m} $ that's induced by the fact that we cannot run randomized experiments for every intervention at every time period. 

\subsection{Bayesian Adaptive design}\label{sec:adaptive_ds}
We propose a Bayesian adaptive design framework for selecting interventions to be randomized in the RCT study.  The goal is to sequentially identify the set of interventions $\bS_{m}$ at each stage such that the cumulative risk is minimized. A key challenge is that the next-stage risk function requires the plug-in covariance estimates of the RCT estimates, $\widehat{\Upsilon}_{m+1}$, which cannot be directly estimated since future randomized data has not yet been observed. To address this challenge, we introduce a Bayesian structure on $\bUpsilon$, the asymptotic covariance of randomized treatment effects. We start with the case where $|\bS_m|=1$ and denote $R_{m+1}( k):=\E\left[\widehat{\mathfrak{R}}_{m+1}\mid S_{m+1} = k, \mathcal{H}_m\right]$ as the expected next-stage risk conditional on randomizing intervention $k$ at round $m+1$, where $\mathcal{H}_m$ denote the historical information up to round $m$.

Recall that $\Upsilon_{jj'}=0$ for all $j\neq j'$ as different interventions are estimated using independent data. We assume a Bayesian hierarchical structure on the diagonal elements of $\bUpsilon$:
\begin{align*}
\Upsilon^{jj} &\sim \text{InvGamma}(\alpha,\beta^j),\quad \beta^j \sim \text{Gamma}(\eta_0,\lambda_0),
\end{align*}
where $\alpha$, $\eta_0$, and $\lambda_0$ are user-specified hyperparameters. {In practice, the hyperparameters can be chosen to center the prior distribution around the estimated asymptotic variances of $\widehat{\tau}^{j}_{\mathcal{R}}$, where $j\in{\bS}_1$ are the interventions randomized in the first round, while accommodating uncertainty through a larger variance. Let $r_m^j = \sum_{\nu=1}^m L_{\nu}\mathbf{1}_{j\in\bS_\nu}$ be the total sample size of RCT studies in which intervention $j$ was randomized by the end of round $m$. A choice of hyperparameters is to set $\eta_0 /\lambda_0\approx(\alpha-1)|{\bS}_1|^{-1}\sum_{j\in{\bS}_1}r^j_{1}\widehat{\Upsilon}^{jj}_1$, so that the prior mean of $\Upsilon^{jj}$ is approximately the average of observed asymptotic variances of randomized treatment effects.  Increasing $\alpha$ or $\lambda_0$ can lead to stronger priors, and vice versa.}

{A key challenge of conventional Bayesian updating in our framework is that it can fail to incorporate full information from new observations. Unlike usual posterior updates where new information consists of $n$ newly observed data points, in our context, the new information is a single asymptotic variance estimate $\widehat{\Upsilon}^{jj}_{m}$ for the $j$-th treatment effect derived from $r^j_m$ samples. The critical point is that $r^j_m\widehat{\Upsilon}^{jj}_{m}$ gives a more accurate estimate of $\Upsilon^{jj}$ through the additional $r^j_m-r^j_{m-1}$ samples collected. To properly incorporate this information, we need  to consider information in two forms: the plug-in estimate $\widehat{\Upsilon}^{jj}{m}$ and the sample size $r^j_m$ used to obtain it.  Since conventional Bayesian updating would lead to a posterior $\beta^j\mid\widehat{\Upsilon}^{jj}_{m}\sim\text{Gamma}(\eta_0+\alpha,\lambda_0+1/\widehat{\Upsilon}^{jj}_{m})$ that treats the new variance estimate as a single data point rather than recognizing that it contains information from $r^j_m$ data points, we propose updating the posterior of $\beta^j$ as follows: \[\beta^j\mid\widehat{\Upsilon}^{jj}_{m},r^j_m\sim\text{Gamma}(\eta_0+r^j_{m}\alpha, \lambda_0+1/\widehat{\Upsilon}^{jj}_{m}).\] This approach allows us to incorporate both the contribution of the $r^j_m$ samples used to estimate $\widehat{\tau}^j_{\mathcal{R},m}$ and its estimated asymptotic variance. The intuition is straightforward: we account for the increasing sample size by updating $\eta_0$ to $\eta_0+r^j_m\alpha$, similar to how we would update $\eta_0$ to $\eta_0+n\alpha$ in the regular case where $n$ additional observations are directly included in the likelihood. The estimate of interest is then incorporated in the updated scale parameter $\lambda_0+1/\widehat{\Upsilon}^{jj}_m$ as usual.}

{Posterior sampling is carried out as follows. After observing new estimates $\widehat{\Upsilon}_{m}^{jj}$ for some $j\in\mathbf{S}_m$, we estimate the next-stage variances $\Upsilon_{m+1}^{jj}$ for all $j\in\{1,...,J\}$ in two stages. First, we sample the posterior predictive of ${\Upsilon}^{jj}$ conditional on $\widehat{\Upsilon}_{m}^{jj}$ and $r_m^j$ (see Lines 1-4 of Algorithm~\ref{alg:TS}). Then, we scale these posterior predictive means by the appropriate sample sizes—either $r_m^j+L_{m+1}$ if intervention $j$ is selected in round $m+1$, or $r_m^j$ if not selected—to obtain the plug-in variance estimates for the next stage (Line 6 of Algorithm~\ref{alg:TS}). To ensure sufficient exploration during the process, we adopt Thompson sampling (TS) for selecting optimal designs. Algorithm~\ref{alg:TS} summarizes the adaptive design framework. }

\begin{algorithm}
\caption{Bayesian adaptive design via Thompson Sampling}
\label{alg:TS}
\begin{algorithmic}[0]
\For{$j = 1$ to $J$}
    \State Sample ${\beta}^j \mid \widehat{\Upsilon}_{m}^{jj} \sim \text{Gamma}(a_0 + r^j_{m}\alpha,\; \lambda_0+\mathbf{1}_{r^j_{m}>0}\widehat{\Upsilon}_{m}^{-jj})$
    \State Sample ${\Upsilon}^{jj} \mid \widehat{\Upsilon}_{m}^{jj} \sim \text{InvGamma}(\alpha,\; {\beta}^j)$
\EndFor
\Statex
\For{$k = 1$ to $J$}
    \State Calculate $\widehat{\bSigma}_{m+1}(k)$ from \eqref{eq:Sigma} using 
        $\{{\Upsilon}_{m+1}^{jj} = ({r_m^j+L_{m+1}\mathbf{1}_{j=k}})^{-1}{\Upsilon}^{jj}$ for all $j=1,\ldots,J\}$
    \State Calculate $R_{m+1}(k) = \widehat{\mathfrak{R}}_{m+1}(\bD,\; \widehat{\bSigma}_{m+1}(k),\; g^{\widehat{\lambda}_m})$
\EndFor
\Statex
\State \Return $S_{m+1} \in \arg\min_k \{R_{m+1}(k)\}$
\end{algorithmic}
\end{algorithm}
While we illustrate the algorithm assuming $|\bS_{m+1}|=1$, it can be easily extended to the case where $|\bS_{m+1}|=n$ by selecting the interventions corresponding to the $n$ smallest $\{R_{m+1}(k)\}_{k=1}^J$. In Proposition~\ref{prop:inf_many}, 
we show that under algorithm~\ref{alg:TS}, each intervention will be selected infinitely many times as $m\to\infty$, which assures the asymptotic behavior of the proposed shrinkage estimator.
\begin{prop}\label{prop:inf_many}
Let $\{\mathbf{S}_m\}$ be a sequence of actions selected according to Algorithm~\ref{alg:TS}. Suppose $\text{Var}\{R_m(k)\}<C_0$ for some $C_0>0$ for all $m,k$. Then, there exist hyperparameters $\alpha,\eta_0,\lambda_0>0$ such that the following properties hold:
\begin{enumerate}[label=(\roman*)]
    \item Each intervention $j$ will be selected infinitely many times as $m\to\infty$, i.e. $r_m^j\to\infty$ almost surely as $m\to\infty$ for all $j$.
\item There exists $M\in\mathbb{N}$ such that
$\sup_{k\in\{1,...,J\}}{R}_{m}(k) - \inf_{k\in\{1,...,J\}}{R}_{m}(k)\leq1$
for all $m\geq M$.
\end{enumerate}
\end{prop}

We then establish a regret bound for the proposed TS algorithm, which is a direct result from \citet{russo2014learning}.
\[
\E\left(\sum_{m=1}^M \widehat{\mathfrak{R}}_{m}\right)=O(M^{1/2}).
\]

    


\section{Simulation Studies}\label{sec:simu}
We evaluate the performance of the proposed optimal shrinkage estimators and adaptive design strategies in this section. We consider $J=100$ interventions with oracle treatment effects $\tau^{j*}=\mathbf{1}_{j\leq 50}-\mathbf{1}_{j>50}$. The outcome model for both observational and RCT studies is:
\begin{align*}\label{eq:outmodel}
    Y_{i,m} &= \sum_{j=1}^n A^j_{i,m}\tau^{j*} + h(\bX_{i,m})^\T\bbeta + \bU_{i,m}^\T\balpha+\sum_{j=1}^n A_{i,m}^j\epsilon^j_{1,i,m}+(1-A_{i,m}^j)\epsilon_{0,i,m},
\end{align*}
where $\epsilon^j_{1,i,m}\overset{iid}{\sim}N(0,\mathbf{1}_{j>50}+0.1\mathbf{1}_{j\leq50})$, $\epsilon_{0,i,m}\overset{iid}{\sim}N(0,0.1)$, and $h(\cdot):\mathcal{X}\to\mathbb{R}^p$ is a mapping function defined later. This setting mimics scenarios with heterogeneous variances across interventions. Regression parameters are set to be $\beta_k = \mathbf{1}_p$  and $\alpha_k=0.5\mathbf{1}_J$. For notation simplicity, we drop the subscripts $(i,m)$ in the rest of the section. We generate $\bX\overset{iid}{\sim}N_{5}(\mathbf{0},\bSigma)$, where $\bSigma_{kk}=1$,  $P(\Sigma_{kk'}=0.15)=0.7$ and $P(\Sigma_{kk'}=0)=0.3$ for all $k\neq k'$. The  $j-$th intervention attributes are sampled as $\bV^j\overset{iid}{\sim} N_3
\left(\mathbf{0}_3,\mathbf{I}_3\right)$. For the observational study, 
we assume that the $j$-th intervention assignment follows a Bernoulli distribution with probability $P(A^j=1\mid \bX, U^j)=1/\{1+\exp(-f(\bX)^\T\bgamma+2U^j)\},$ where $f(\cdot):\mathcal{X}\to\mathbb{R}^{p_A}$ is a mapping function and $\bgamma=0.5\cdot\mathbf{1}_{p_A}$. The unobserved confounder $\bU$ is generated from a zero-mean Gaussian process model with the Mat\`ern-1/2 kernel function. 

 For RCT data, we sample $W_{i,m}$, the intervention to be randomized in experiment $i$ at round $m$, randomly from $\bS_m$ as described in Section~\ref{sec:setup}. The $j$-th intervention assignment in the RCT study thus follows a Bernoulli distribution with $P(A^j=1\mid\bX,U^j, W)=0.5\mathbf{1}_{W=j}+1/\{1+\exp(-f(\bX)^\T\bgamma+2U^j)\}\mathbf{1}_{W\neq  j}$, where $\bX$ and $\bU$ are generated the same as in the observational study. In this numerical experiment, we set $f(\mathbf{x})=\mathbf{x}$ and $h(\mathbf{x})=(x_1,x_2,x_3,x_4,x_5,x_1x_2,x_1x_3,x_1x_4,x_2x_3,x_2x_4,x_3x_4).$

We consider standard squared error loss with identity weighting matrix $\bD=\mathbf{I}_J/J$, and set the sample sizes of the observational and RCT study to be $N_m=5,000$ and $L_m=2,000$ for all $m$, respectively. We use degree-three spline regressions for the bias model and fit the propensity score model using logistic regression with spline basis functions of context $\bX$. Initially, we randomly select 15 interventions to be evaluated in the RCT study, and select $n=5$ additional interventions for randomized evaluation at each subsequent round. For hyperparameters, we pick $\alpha=5, a_0=10$, and $\lambda_0=0.05$. All experiments are run on 2.5GHz Intel Xeon Platinum 8259CL CPU with 16 vCPUs and 8G RAM.

\begin{figure}[h]
    \centering
    \includegraphics[width=0.49\linewidth]{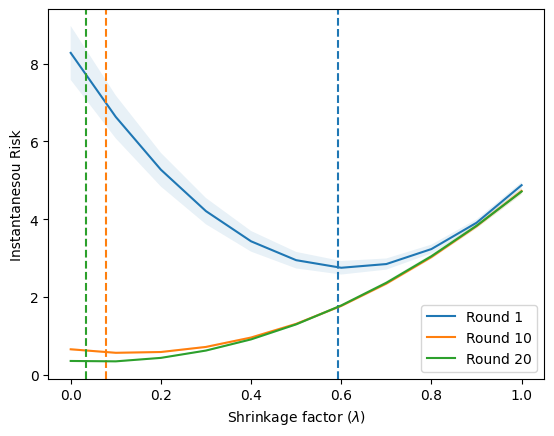}
           \includegraphics[width=0.49\linewidth]{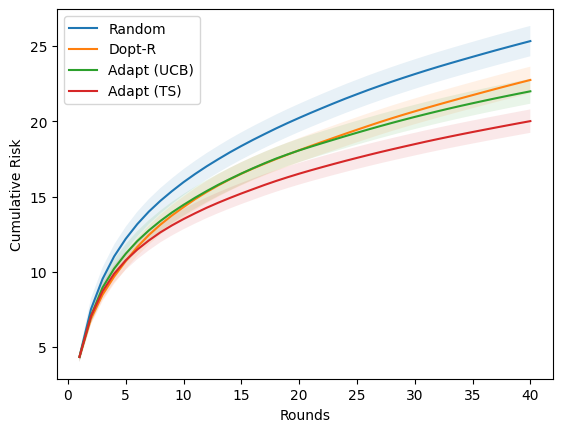} 
    \caption{[Left] Risk of $\widehat{\btau}^\lambda$ using different shrinkage factors $\lambda\in[0,1]$ at round 1, 10, and 20. Dotted vertical lines indicate the estimated optimal shrinkage factor ($\widehat{\lambda}^{*}$); [Right] Cumulative risk (log scale) of various designs and sampling methods using the optimal shrinkage estimator.}
    \label{fig:simu}
\end{figure}

We benchmark our approach against several state-of-the-arts methods: (1) random sampling (Random), (2) D-optimal designs followed by random sampling (Dopt-R), and (3) proposed adaptive design with upper-confidence-bound algorithm (UCB). The Dopt-R approach uses D-optimal sequential designs until each intervention has been selected once, after which it shifts to random sampling since traditional D-optimal designs are not designed to repeatedly sample the same interventions.
Figure~\ref{fig:simu} shows the performance of the optimal shrinkage estimator and adaptive design. First, we validate that the estimated optimal shrinkage factor (dotted line) accurately identifies the risk-minimizing value, with $\widehat{\lambda}^*$ decreasing over rounds as RCT confidence grows. This matches our expectation, as larger RCT samples provide greater confidence in bias estimates, naturally shifting preference toward the fully-debiased estimator $(\lambda=0)$. The right plot of Figure~\ref{fig:simu} presents the cumulative risk across different designs and sampling methods. While D-optimal designs demonstrate strong initial performance, their efficiency diminishes after exhausting unique intervention selections. In contrast, our proposed adaptive designs achieve substantially lower cumulative risk compared to both random and D-optimal approaches, with the Thompson Sampling implementation showing particularly strong performance.

\section{Application}\label{sec:da}

Finally, we evaluate our proposed shrinkage estimators and adaptive sampling method using Amazon's advertising campaign data. In collaboration with Amazon Ads ECON team, we analyze 2,583 campaigns implemented in 2024. Figure~\ref{fig:sponsordisplay} provides an example of the campaigns of interest. The treatment effect of interests in this study is the change in product page view rate. RCT campaign-level treatment effects and their standard errors are estimated through a ghost-ads infrastructure \citep[]{johnson2017ghost}, while observational campaign-level effects are derived by aggregating impression-level effects from a DNN-uplift model developed by the Ads ECON team. Based on domain expertise, we identify 18 intervention attributes, including campaign types, media types, targeted product price, and targeted product views. The bias model is fit using spline regression of degree three.
\begin{figure}[h]
    \centering
    \includegraphics[width=1\linewidth]{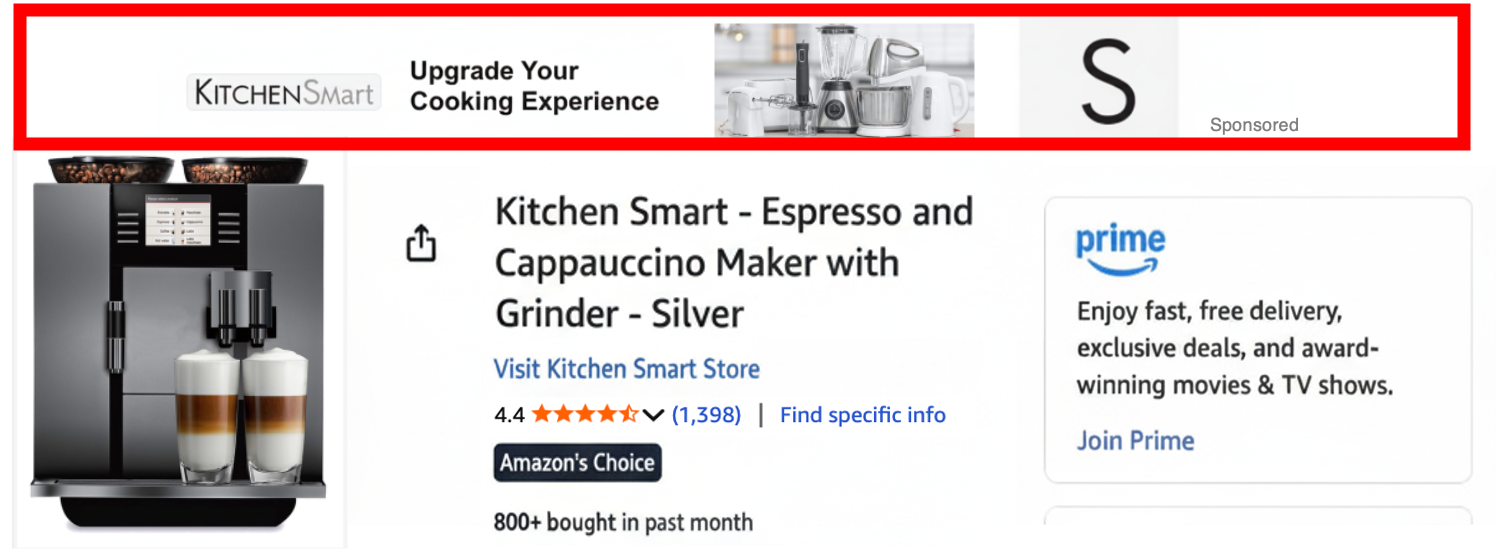}
    \caption{ An example of the advertising campaigns of interest, highlighted by the red box. Brands and merchants are fictional.}\label{fig:sponsordisplay}
\end{figure}

In implementing our adaptive design framework, we initialize with data from 500 RCT campaigns. At each round, we select 100 campaigns for RCT evaluation in each of the 20 rounds. Since it is an offline analysis and it is infeasible to update RCT effects without additional experimental data, choosing previously-sampled campaigns for RCT evaluation will not provide new information in this analysis. Thus, we focus on change without replacement sampling scheme in this real data application.  We compare our proposed method against random sampling, {which represents the standard approach widely used in practice for this type of experimental design task.}

Figure~\ref{fig:DA} shows the estimated optimal shrinkage factors and RCT costs versus instantaneous risk difference. Since oracle treatment effects for all treatments are unknown, we use risk difference as a proxy metric to evaluate model performance. We define risk difference as the difference between our shrinkage estimator $\widehat{\btau}^{\lambda}_m(\overline{\mathbf{S}}_m)$ based on RCT data from selected campaigns and an optimal benchmark estimator that uses RCT data from all 2,583 campaigns. First, we observe that the estimated optimal shrinkage factors derived using our sampling method is significantly lower than the of the random sampling, demonstrating that our approach identifies informative interventions for randomization more effectively. As $m$ increases, both estimated optimal shrinkage factors converge, which aligns with our expectation since by round 20, both methods will have leveraged nearly all 2,583 randomized campaigns. On the other hand, The right plot of Figure~\ref{fig:DA} reveals substantial cost savings that can be benefited from our method: the proposed sampling method achieves the same performance level that random sampling achieves (instantaneous risk: 0.02) while reducing the costs by approximately 50\%. While these cost estimates are proportional, they indicate the significant resource savings potential of our adaptive approach. Additional cost savings are possible due to opportunity costs of impressions and unrealized incremental conversions.
\begin{figure}[h]
    \centering
    \includegraphics[width=0.49\linewidth]{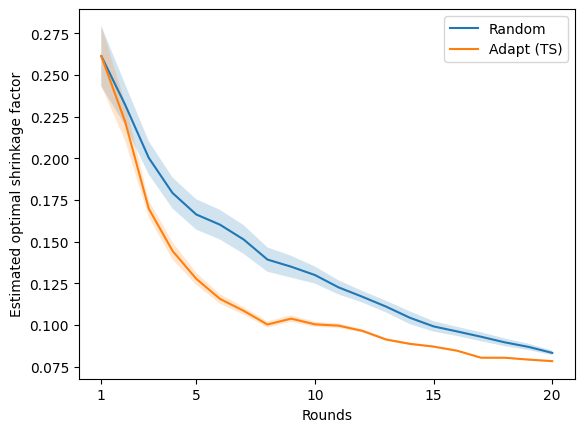}
    \includegraphics[width=0.49\linewidth]{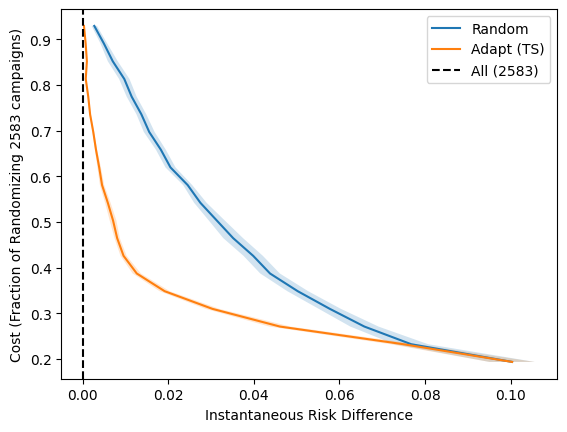}
    \caption{[Left] Estimated optimal shrinkage factors; [Right] Estimated RCT costs versus instantaneous risk. The risk difference is defined as the the difference between the risk of $\widehat{\btau}_m^{\lambda}(\overline{\bS}_m)$ and the risk of the shrinkage estimator using RCT data of all 2,583 campaigns. Shaded areas represent one standard error. Dotted line indicates the risk using all RCT estimates at the final round.}
    \label{fig:DA}
\end{figure}

\section{Discussion}\label{sec:discussion}
We proposed an optimal shrinkage estimator and adaptive sampling framework for 
multiple campaign effect estimation in large marketplaces, where evaluating all 
campaigns through RCTs is impractical. Our Bayesian adaptive design framework 
maximizes resource efficiency by judiciously selecting which campaigns to 
evaluate through RCT at each time point. Application of our method to Amazon's 
advertising campaign data demonstrates substantial efficiency gains, reducing 
costs by approximately 50\% compared to random sampling. To the best of our 
knowledge, this paper is the first to develop: 1) an optimal shrinkage estimator 
that handles missing randomized data, and 2) a sequential sampling algorithm for 
efficient implementation of randomized experiments.

{There are several future directions worth investigating. First, allowing different shrinkage factors across interventions could improve performance by applying more shrinkage to where RCT sample variance is larger. Second, our current method assumes the de-biased model is correctly specified. While we adopt flexible feature mapping to capture nonlinear bias, the model may yield misleading results if important features are omitted. Developing a feature selection framework for the de-biased model and assessing model robustness under misspecification would therefore be valuable extensions. Finally, adapting our sampling method to account for heterogeneous intervention costs would broaden its practical applications across different marketplace settings.}

\section*{Acknowledgement}
The authors Yen-Chun Liu, Alexander Volfovsky, and Eric Laber acknowledge support from Amazon.

\spacingset{1.0}
\setlength{\bibsep}{0pt plus 0.3ex}
\footnotesize
\bibliographystyle{apalike}
\bibliography{biometrika_reference}

\begin{thebibliography}{20}
\expandafter\ifx\csname natexlab\endcsname\relax\def\natexlab#1{#1}\fi

\bibitem[{Chen et~al.(2015)Chen, Owen \& Shi}]{chen2015data}
\textsc{Chen, A.}, \textsc{Owen, A.~B.} \& \textsc{Shi, M.} (2015).
\newblock Data enriched linear regression.
\newblock \textit{Electronic journal of statistics} \textbf{9}, 1078--1112.

\bibitem[{Colnet et~al.(2020)Colnet, Mayer, Chen, Dieng, Li, Varoquaux, Vert,
  Josse \& Yang}]{colnet2020causal}
\textsc{Colnet, B.}, \textsc{Mayer, I.}, \textsc{Chen, G.}, \textsc{Dieng, A.},
  \textsc{Li, R.}, \textsc{Varoquaux, G.}, \textsc{Vert, J.-P.}, \textsc{Josse,
  J.} \& \textsc{Yang, S.} (2020).
\newblock Causal inference methods for combining randomized trials and
  observational studies: a review.
\newblock \textit{arXiv preprint arXiv:2011.08047} .

\bibitem[{Dawson \& Lavori(2008)}]{dawson2008sequential}
\textsc{Dawson, R.} \& \textsc{Lavori, P.~W.} (2008).
\newblock Sequential causal inference: Application to randomized trials of
  adaptive treatment strategies.
\newblock \textit{Statistics in Medicine} \textbf{27}, 1626--1645.

\bibitem[{Degtiar \& Rose(2021)}]{degtiar2021review}
\textsc{Degtiar, I.} \& \textsc{Rose, S.} (2021).
\newblock A review of generalizability and transportability.
\newblock \textit{arXiv preprint arXiv:2102.11904} .

\bibitem[{Fourdrinier et~al.(2018)Fourdrinier, Strawderman \&
  Wells}]{fourdrinier2018shrinkage}
\textsc{Fourdrinier, D.}, \textsc{Strawderman, W.~E.} \& \textsc{Wells, M.~T.}
  (2018).
\newblock \textit{Shrinkage estimation}.
\newblock Springer.

\bibitem[{Funk et~al.(2011)Funk, Westreich, Wiesen, St{\"u}rmer, Brookhart \&
  Davidian}]{funk2011doubly}
\textsc{Funk, M.~J.}, \textsc{Westreich, D.}, \textsc{Wiesen, C.},
  \textsc{St{\"u}rmer, T.}, \textsc{Brookhart, M.~A.} \& \textsc{Davidian, M.}
  (2011).
\newblock Doubly robust estimation of causal effects.
\newblock \textit{American journal of epidemiology} \textbf{173}, 761--767.

\bibitem[{Gui(2020)}]{gui2020combining}
\textsc{Gui, G.} (2020).
\newblock Combining observational and experimental data using first-stage
  covariates.
\newblock \textit{arXiv preprint arXiv:2010.05117} .

\bibitem[{Johnson et~al.(2017)Johnson, Lewis \& Nubbemeyer}]{johnson2017ghost}
\textsc{Johnson, G.~A.}, \textsc{Lewis, R.~A.} \& \textsc{Nubbemeyer, E.~I.}
  (2017).
\newblock Ghost ads: Improving the economics of measuring online ad
  effectiveness.
\newblock \textit{Journal of Marketing Research} \textbf{54}, 867--884.

\bibitem[{Kallus et~al.(2018)Kallus, Puli \& Shalit}]{kallus2018removing}
\textsc{Kallus, N.}, \textsc{Puli, A.~M.} \& \textsc{Shalit, U.} (2018).
\newblock Removing hidden confounding by experimental grounding.
\newblock \textit{Advances in neural information processing systems}
  \textbf{31}.

\bibitem[{Kaptein \& Eckles(2012)}]{kaptein2012heterogeneity}
\textsc{Kaptein, M.} \& \textsc{Eckles, D.} (2012).
\newblock Heterogeneity in the effects of online persuasion.
\newblock \textit{Journal of Interactive Marketing} \textbf{26}, 176--188.

\bibitem[{Lewis \& Rao(2015)}]{lewis2015unfavorable}
\textsc{Lewis, R.~A.} \& \textsc{Rao, J.~M.} (2015).
\newblock The unfavorable economics of measuring the returns to advertising.
\newblock \textit{The Quarterly Journal of Economics} \textbf{130}, 1941--1973.

\bibitem[{Moodie et~al.(2007)Moodie, Richardson \&
  Stephens}]{moodie2007demystifying}
\textsc{Moodie, E.~E.}, \textsc{Richardson, T.~S.} \& \textsc{Stephens, D.~A.}
  (2007).
\newblock Demystifying optimal dynamic treatment regimes.
\newblock \textit{Biometrics} \textbf{63}, 447--455.

\bibitem[{Niemiro(1992)}]{niemiro1992asymptotics}
\textsc{Niemiro, W.} (1992).
\newblock Asymptotics for m-estimators defined by convex minimization.
\newblock \textit{The Annals of Statistics} , 1514--1533.

\bibitem[{Rosenman et~al.(2020)Rosenman, Basse, Owen \&
  Baiocchi}]{rosenman2020combining}
\textsc{Rosenman, E.}, \textsc{Basse, G.}, \textsc{Owen, A.} \&
  \textsc{Baiocchi, M.} (2020).
\newblock Combining observational and experimental datasets using shrinkage
  estimators.
\newblock \textit{arXiv preprint arXiv:2002.06708} .

\bibitem[{Russo \& Van~Roy(2014)}]{russo2014learning}
\textsc{Russo, D.} \& \textsc{Van~Roy, B.} (2014).
\newblock Learning to optimize via information-directed sampling.
\newblock In \textit{Advances in Neural Information Processing Systems}.

\bibitem[{Strawderman(2003)}]{strawderman2003minimax}
\textsc{Strawderman, W.~E.} (2003).
\newblock On minimax estimation of a normal mean vector for general quadratic
  loss.
\newblock \textit{Lecture Notes-Monograph Series} , 3--14.

\bibitem[{Toth et~al.(2022)Toth, Lorch, Knoll, Krause, Pernkopf, Peharz \&
  Von~K{\"u}gelgen}]{toth2022active}
\textsc{Toth, C.}, \textsc{Lorch, L.}, \textsc{Knoll, C.}, \textsc{Krause, A.},
  \textsc{Pernkopf, F.}, \textsc{Peharz, R.} \& \textsc{Von~K{\"u}gelgen, J.}
  (2022).
\newblock Active bayesian causal inference.
\newblock \textit{Advances in Neural Information Processing Systems}
  \textbf{35}, 16261--16275.

\bibitem[{Vermeulen \& Vansteelandt(2015)}]{vermeulen2015bias}
\textsc{Vermeulen, K.} \& \textsc{Vansteelandt, S.} (2015).
\newblock Bias-reduced doubly robust estimation.
\newblock \textit{Journal of the American Statistical Association}
  \textbf{110}, 1024--1036.

\bibitem[{Yang et~al.(2020)Yang, Zeng \& Wang}]{yang2020improved}
\textsc{Yang, S.}, \textsc{Zeng, D.} \& \textsc{Wang, X.} (2020).
\newblock Improved inference for heterogeneous treatment effects using
  real-world data subject to hidden confounding.
\newblock \textit{arXiv preprint arXiv:2007.12922} .

\bibitem[{Zhang et~al.(2023)Zhang, Cammarata, Squires, Sapsis \&
  Uhler}]{zhang2023active}
\textsc{Zhang, J.}, \textsc{Cammarata, L.}, \textsc{Squires, C.},
  \textsc{Sapsis, T.~P.} \& \textsc{Uhler, C.} (2023).
\newblock Active learning for optimal intervention design in causal models.
\newblock \textit{Nature Machine Intelligence} \textbf{5}, 1066--1075.

\end{thebibliography}

\newpage
\renewcommand{\appendixpagename}{Appendix}

\begin{appendices}
\section{Proof of Lemmas and Theorem}
\subsection*{Proof of Lemma 3.1}
The following is a well-known result in linear algebra.  We include it here along with a proof for completeness.
\begin{lemma}
Suppose that $\bSigma$ and $\bQ$ are symmetric and  strictly positive definite.  T
hen there exists a non-singular matrix $\bOmega$ such that
$\bOmega\bQ \bOmega^{\T} = I$ and $(\bOmega^{\T})^{-1}\bSigma \bOmega^{-1} = \bLambda$, where
$\bLambda$ is a diagonal matrix.   
\end{lemma}
\begin{proof}
Write $\bSigma = \bC^{\T}\bC$, then as $\bC^{-\T}\bD^{-1}\bC^{-1}$ is also symmetric, there
exists orthogonal matrix $\bCalO$ such that $\bCalO^{\T}\left\lbrace 
\bC^{-\T}\bD^{-1}\bC^{-1}
\right\rbrace \bCalO = \bLambda^{-1}$ where $\bLambda^{-1}$ is diagonal.
Hence, inverting both sides, it follows that 
$\bCalO^{\T}\bC\bD\bC^{\T}\bCalO = \bLambda$, where
we have used the fact that $\mathcal{O}^{-1} = \mathcal{O}^{\T}.$
Define $\bOmega = \bCalO^{\T}\bC^{-\T}$.  It easily verified 
that $\bOmega$ satisfies the desired properties, as 
\begin{eqnarray*}
    \bOmega^{-\T}\bD \bOmega^{-1} &=& 
    \left(
    \bCalO^{\T}\bC^{-\T}
    \right)^{-\T}\bSigma 
    \left(
    \bCalO^{\T}\bC^{-\T}
    \right)^{-1} \\ 
    &=& \bCalO^{\T}\bC^{\T}\bD \bC \bCalO \\
    &=& \bLambda,
\end{eqnarray*}
and 
\begin{eqnarray*}
\bOmega \bSigma \bOmega^{\T} &=& \bCalO^{\T}\bC^{-\T}\bSigma \bC^{-1}\bCalO \\
&=& \bCalO^{\T}\bC^{-\T}\bC^{\T}\bC \bC^{-1}\bCalO \\ 
&=& \bCalO^{\T}\bCalO \\
&=& I. 
\end{eqnarray*}
\end{proof}

\begin{proof}[Proof of Lemma 3.1]
Let $\bSigma = \bC^{\T}\bC$ be the Cholesky decomposition 
of $\bSigma$ and let $\bCalO$ be an orthogonal matrix
such that $\bCalO^{\T}\bC^{\T}\bD\bC\bCalO = \bLambda = \mathrm{diag}(\lambda_1,\ldots, \lambda_p)$,
and define $\bOmega = \mathcal{O}^{\T}\bC^{-\T}$, where $\mathbf{C}^{-\T} = (\mathbf{C}^{\T})^{-1}.$ 
The risk associated with $\kappa(\bZ, \bY)$ is 
\begin{eqnarray*}
    \mathfrak{R}\left\lbrace \bD,\bSigma,g
    \right\rbrace &=& 
    \mathbb{E}\mathcal{L}_{\bD} \left\lbrace
        \kappa(\bZ, \bY), \bmu
    \right\rbrace \\[3pt]
    &=& \left\lbrace
    \kappa(\bZ, \bY) - \bmu
    \right\rbrace^{\T}\bD 
    \left\lbrace
    \kappa(\bZ, \bY) - \bmu    
    \right\rbrace \\[3pt]
&=& \left\lbrace 
\bOmega \kappa(\bZ, \bY) - \bOmega \bmu
\right\rbrace
\bOmega^{-\T} \bD \bOmega^{-1}
\left\lbrace^{\T} 
\bOmega \kappa(\bZ, \bY) - \bOmega \bmu
\right\rbrace \\[3pt]
&=&
\left\lbrace 
\bOmega \kappa(\bOmega^{-1}\widetilde{\bZ}, \bOmega^{-1}\widetilde{\bY}) - \widetilde{\bmu}
\right\rbrace^{\T}
\bLambda 
\left\lbrace 
\bOmega \kappa(\bOmega^{-1}\widetilde{\bZ}, \bOmega^{-1}\widetilde{\bY}) - \widetilde{\bmu}
\right\rbrace^{\T},
\end{eqnarray*}
where $\widetilde{\bZ} = \bOmega \bZ$, $\widetilde{\bY} = \bOmega Y$, and $\widetilde{\bmu} = \bOmega \bmu$.
Note that 
\begin{eqnarray*}
    \bOmega \kappa(\bOmega^{-1}\widetilde{\bZ}, \bY) &=& 
    \widetilde{\bZ} + \bOmega \bSigma g(\bOmega^{-1}\widetilde{\bZ}, \bOmega^{-1}\widetilde{\bY}) \\[3pt] 
    &=& \widetilde{\bZ} + \widetilde{g}(\widetilde{\bZ}, \widetilde{\bY}),
\end{eqnarray*}
where $\widetilde{g}(\widetilde{\bZ}, \widetilde{\bY}) = \bOmega\bSigma g(\bOmega^{-1}\widetilde{\bZ}, \bOmega^{-1}\widetilde{\bY})$. 
Because $\widetilde{\bZ} \sim \mathcal{N}(\widetilde{\bmu}, \mathbf{I})$,  
if we define 
$\widetilde{\kappa}(\widetilde{\bZ}, \widetilde{\bY}) = \bOmega \kappa(\bOmega^{-1}\widetilde{\bZ}, \bOmega^{-1}\widetilde{\bY})$,
it follows that $  \mathfrak{R}\left\lbrace \bD,\bSigma,g
    \right\rbrace = \mathfrak{R}
\left\lbrace \bLambda,\mathbf{I},\widetilde{g} \right\rbrace$;
this is slight extension of Theorem 3.13 in \citet{fourdrinier2018shrinkage}.  

Thus, applying Theorem 1 of \citet{rosenman2020combining}, we have 
\begin{equation*}
    \mathfrak{R}
\left\lbrace \bLambda,\mathbf{I},\widetilde{g} \right\rbrace = 
    \frac{1}{p}\mathrm{tr}(\bLambda) + \frac{1}{p}\mathbb{E}\Bigg[
    \sum_{j=1}^p\lambda_j 
    \widetilde{g}_j^2(\widetilde{\bZ}, \widetilde{\bY}) + 2
    \frac{
    \partial \widetilde{g}_j(\widetilde{\bZ}, \widetilde{\bY}) 
    }{
    \partial Z_j
    }
    \Bigg],
\end{equation*}
where $\bLambda = \mathrm{diag}(\lambda_1,\ldots, \lambda_p)$.  
\end{proof}

\subsection*{Proof of Lemma 3.2}
Let $\overline{N}_m = \sum_{\nu=1}^m N_m$ and $\overline{L}_m = \sum_{\nu=1}^m L_m$. Suppose $\sqrt{\overline{N}_m}\widehat{\btau}_{\mathcal{O},m}\overset{d}{\to} N(\btau_{\mathcal{O}}^*,\bGamma)$ and $\sqrt{\overline{L}_m}\widehat{\btau}_{\mathcal{R},m}\overset{d}{\to} N(\btau^*,\bUpsilon)$. The least square estimator $\widehat{\btheta}_m$ can be represented as
\[
\widehat{\btheta}_m = \left(\bPsi_{\overline{\mathbf{S}}_m}^\T\bPsi_{\overline{\mathbf{S}}_m}\right)^{-1}\widetilde{\bPsi}_{\overline{\mathbf{S}}_m}^\T \left\{\widehat{\btau}_{\mathcal{O},m}(\overline{\mathbf{S}}_m)-\widehat{\btau}_{\mathcal{R},m}(\overline{\mathbf{S}}_m)\right\},
\]
where 
$\bPsi_{\overline{\mathbf{S}}_m}\in\mathbb{R}^{|\{j\mid j\in\cup_{\nu=1}^m\overline{\mathbf{S}}_\nu\}|\times p_v}$ be a submatrix of $\bPsi$ consisting of the rows indexed by unique elements in $\cup_{\nu=1}^m{\mathbf{S}}_\nu$, and let
$\widetilde{\bPsi}_{\overline{\mathbf{S}}_m}\in\mathbb{R}^{{J}\times p_V}$ 
be $\bPsi$ with entries not indexed by $\cup_{\nu=1}^m{\mathbf{S}}_\nu$ replaced with zeros. Assume each intervention $j$ will be selected infinitely many times as $m\to\infty$, then
\begin{align*}
 &\frac{1}{\sqrt{\overline{N}_m^{-1}+\overline{L}_m^{-1}}\sqrt{\overline{N}_m}} \sqrt{\overline{N}_m}\widehat{\btau}_{\mathcal{O},m}\overset{d}{\to}N\left\{\widehat{\btau}^*_{\mathcal{O}},(1-\rho)\bGamma\right\},\text{ and}\\
 &\frac{1}{\sqrt{\overline{N}_m^{-1}+\overline{L}_m^{-1}}\sqrt{\overline{L}_m}} \sqrt{\overline{L}_m}\widehat{\btau}_{\mathcal{R},m}\overset{d}{\to}N\left(\widehat{\btau}^*,\rho\bUpsilon\right),
\end{align*}
and thus
\[\frac{1}{\sqrt{\overline{N}_m^{-1}+\overline{L}_m^{-1}}}(\bPsi\widehat{\btheta}_m-\bPsi\btheta^*)\overset{d}{\to}N_{p_V}\left[\mathbf{0},\bH\left\{(1-\rho)\bGamma + \rho\Upsilon \right\}\bH^\top\right],\]
if $\rho = \lim \overline{N}_m/(\overline{N}_m+\overline{L}_m)\in(0,1)$ and  $\bH=\bPsi(\bPsi^\T\bPsi)^{-1}{\bPsi}^\T.$ Since 
\[
\widehat{\btau}_{\mathcal{O},m} - \bPsi\widehat{\btheta}_m = (\mathbf{I}-\mathbf{H})\widehat{\btau}_{\mathcal{O},m}  + \mathbf{H} \widehat{\btau}_{\mathcal{R},m},
\]
we have
\[
\frac{1}{\sqrt{\overline{N}_m^{-1}+{\overline{L_m}}^{-1}}} \left(\widehat{\btau}_{\mathcal{O},m} - \bPsi\widehat{\btheta}_m-\btau^*\right)\overset{d}{\to} \mathcal{N}\left\{0, (1-\rho)(\mathbf{I}-\bH)\bGamma(\mathbf{I}-\bH)^\T + \rho\bH \Upsilon \bH^\T\right\},
\]
\subsection*{Proof of the closed form of eURE}
Recall that in Lemma 3.1 we show that \[
\mathfrak{R}\left\lbrace \bD,\bSigma,g
    \right\rbrace=
 \mathfrak{R}
\left\lbrace \bLambda,\mathbf{I},\widetilde{g} \right\rbrace=\frac{1}{p} \operatorname{tr}(\bLambda)+\frac{1}{p} \mathbb{E}\left[\sum_{j=1}^p \lambda_j\left\{ \widetilde{g}_j^2(\widetilde{\boldsymbol{Z}}, \widetilde{\boldsymbol{Y}})+2 \frac{\partial \widetilde{g}_j(\widetilde{\boldsymbol{Z}}, \widetilde{\boldsymbol{Y}})}{\partial \tilde{Z_j}}\right\}\right],
\]
where 
$\tilde{g}(\tilde{\bZ},Y)=\bOmega\bSigma g(\bOmega^{-1}\tilde{\bZ},Y)$, $\tilde{\kappa}(\tilde{\bZ},\tilde{Y})=\tilde{\bZ}+\tilde{g}(\tilde{\bZ},\tilde{Y})=\bOmega \kappa(\bOmega^{-1}\tilde{\bZ},\tilde{Y})$, and $\tilde{\bZ}=\bOmega\bZ.$ Note that elements in $\bLambda=(\lambda_1,\cdots,\lambda_p)^\T$  is different from the shrinkage factor $\lambda$. Let $\bZ = \widehat{\btau}_{\mathcal{O},m} - \Psi\widehat{\btheta}_m$, $\bY=\widehat{\btau}_{\mathcal{O},m}$, and $g^\lambda(\bZ,\bY)=\lambda\bSigma^{-1}(\bY-\bZ)$, so that $\kappa(\bZ,\bY)=\widehat{\btau}_{\mathcal{O},m}-(1-\lambda) \Psi\widehat{\btheta}_m$. We first derive the closed form of $\E \sum_{j=1}^p \lambda_j\left\{\partial \widetilde{g}_j(\widetilde{\boldsymbol{Z}}, \widetilde{\boldsymbol{Y}})/{\partial \tilde{Z_j}}\right\}$. By Stein's lemma, it is equivalent to evaluate
\begin{align*}
\E\left\{\left(\widetilde{\bZ}-\bOmega\btau^*\right)^\T \bLambda\widetilde{g}(\widetilde{\boldsymbol{Z}}, \widetilde{\boldsymbol{Y}})\right\}= \tr\left[\bLambda \text{Cov} \left(\widetilde{\bY}-\widetilde{\bZ},\widetilde{\bZ}\right)\right].
\end{align*}
Define $\bH_m=\bPsi (\bPsi_{\overline{\mathbf{S}}_m}^\T\bPsi_{\overline{\mathbf{S}}_m})^{-1}\widetilde{\bPsi}_{\overline{\mathbf{S}}_m}^\T$.
Note that \begin{align*}
    \widetilde{\bZ}&= \bOmega\left[\widehat{\btau}_{\mathcal{O},m} - \bPsi\{\bPsi_{\overline{\mathbf{S}}_m}^\T \bPsi_{\overline{\mathbf{S}}_m}\}^{-1}\widetilde{\bPsi}_{\overline{\mathbf{S}}_m}^\T(\widehat{{\btau}}_{\mathcal{O},m}-\widehat{{\btau}}_{\mathcal{R},m})\right]\\
    &=\bOmega\widehat{\btau}_{\mathcal{O},m} - \bOmega\bH_m\left\{\widehat{{\btau}}_{\mathcal{O},m}-\widehat{{\btau}}_{\mathcal{R},m}\right\}\\
    &=\bOmega\left(\mathbf{I}-\bH_m\right)\widehat{\btau}_{\mathcal{O},m}+\bOmega\bH_m\widehat{\btau}_{\mathcal{R},m},
\end{align*} 
and $\widetilde{\bY} = \bOmega\widehat{\btau}_{\mathcal{O},m}$. 
Thus,
\begin{align*}
   \tr\left[\bLambda \text{Cov} \left(\widetilde{\bY}-\widetilde{\bZ},\widetilde{\bZ}\right)\right]
    &=\tr\left[\bLambda\left\{ \bOmega\bGamma(\mathbf{I}-\bH_m)\bOmega^\T-\mathbf{I}\right\}\right]\\
    &=\tr\left\{\bLambda\bOmega\bGamma(\mathbf{I}-\bH_m)\bOmega^\T-\bLambda\right\}\\
 (\because \bOmega^{-\T}\bD\bOmega^{-1}=\bLambda)\quad   &=\tr\left\{\bD\bGamma(\mathbf{I}-\bH_m) - \bD\bSigma\right\}.
\end{align*}
Here we use the fact that that $\tr(\bLambda) = \tr(\bCalO^\T\bC\bD\bC^\T\bCalO)=\tr(\bD\bSigma).$
Suppose $\lambda$ is fixed. We therefore show that the risk of $\widehat{\btau}_{\mathcal{O},m}-(1-\lambda)\bPsi\widehat{\btheta}_m$ is
\begin{align*}
        \mathfrak{R}_m(\bD, \bSigma, g^\lambda)&= \frac{1}{p}\mathbb{E}\left[\tr(\bLambda) + \tilde{g}(\tilde{\bZ},Y)^\T\bLambda\tilde{g}(\tilde{\bZ},Y) + 2\lambda\tr\left\{\bD{\bGamma}(\mathbf{I}-\bH_m) - \bD{\bSigma}\right\}\right]\\
         &= \frac{1}{p}\left[\tr(\bLambda) + \lambda^2\{\bOmega(Y-\bZ)\}^\T\bLambda\{\bOmega(Y-\bZ)\}  + 2\lambda\tr\left\{\bD{\bGamma}(\mathbf{I}-\bH_m) - \bD{\bSigma}\right\}\right]\\
        &=\frac{1}{p}\left[\tr(\bD\bSigma) +\lambda^2 \E \left\{\left(\bPsi\widehat{\btheta}_m\right)^\T\bD\bPsi\widehat{\btheta}_m\right\}+2\lambda\tr\left\{\bD\bGamma\left(\mathbf{I}-\bH_m\right)^\T -\bD\bSigma\right\}\right].
\end{align*}
Replacing $\bSigma$ with $\widehat{\bSigma}_m=(\mathbf{I}-\bH_m)\widehat{\bGamma}_m(\mathbf{I}-\bH_m)^\T + \bH_m\widehat{\bUpsilon}_m\bH_m^\T$ and $\bGamma$ with $\widehat{\bGamma}_m$ yields
\[
 \widehat{\mathfrak{R}}_m(\bD, \widehat{\bSigma}_m, g^\lambda)=\frac{1}{p}\left[\tr(\bD\widehat{\bSigma}_m) +\lambda^2 \left(\bPsi\widehat{\btheta}_m\right)^\T\bD\bPsi\widehat{\btheta}_m+2\lambda\tr\left\{\bD\widehat{\bGamma}_m\left(\mathbf{I}-\bH_m\right)^\T -\bD\widehat{\bSigma}_m\right\}\right].
\]

\subsection*{Proof of Theorem 3.1}

Optimal $\widehat{\lambda}_{m}$ exists since $\widehat{\mathfrak{R}}_m(\bD, \widehat{\bSigma}_m, g^\lambda)$ is strictly convex when $\Psi\widehat{\btheta}_m\neq \boldsymbol{0}_p$. Solving $\partial\widehat{\mathfrak{R}}_m(\bD, \widehat{\bSigma}_m, g^\lambda)/\partial \lambda=0$ yields
\[
\widehat{\lambda}_{m} = \frac{\tr\left\{-\bD\widehat{\bGamma}_m(\mathbf{I}-\bH_m) + \bD\widehat{\bSigma}_m\right\}}{(\Psi\widehat{\btheta}_m)^\T\bD\Psi\widehat{\btheta}_m}.
\]
Plugging $\lambda=\widehat{\lambda}_m$ into $\widehat{\mathfrak{R}}_m(\bD, \widehat{\bSigma}_m, g^\lambda)$ and replacing expectation with realizations yield the eURE.


\subsection*{Proof of Proposition 3.1}
Let $E_m(j):=\{\text{intervention } j \text{ is selected at round } m\}$.  For notation simplicity, we denote $E_m$ as $E_m(j)$ and $E_m^c$ as the complement event in the following proof. Before proving Lemma 3.3., we will show the following proposition.
\begin{prop}
For sufficiently large $n$, the following inequality holds:
\[P\left(E_m^c\left| \bigcap_{t=n}^{m-1} E_t^c\right.\right)\leq P\left(E_{m-1}^c\left| \bigcap_{t=n}^{m-2} E_r^c\right.\right).  \]
\end{prop}
\begin{proof}
    Let $c:=(\bPsi\widehat{\btheta}_m)^\T\bD\bPsi\widehat{\btheta}_m$ be a positive constant and define $w_{jj}:= (\bH_m^\T\bD\bH_m)_{jj}$. Recall that $\overline{L}_{m}^j$ is the total sample size used for randomizing $j$ by the end of round $m$. Denote the plug-in variance estimate of the next-stage estimator $\tau^j_{\mathcal{R},m+1}$ where $k$ is selected at round $m+1$ as $\widehat{\Upsilon}_{m+1,jj}(k):=({\overline{L}^j_m+L_{m+1}\mathbf{1}_{k=j}})^{-1}\widehat{\Upsilon}_{jj}$, where $\widehat{\Upsilon}_{jj}\mid \widehat{\Upsilon}_{m,jj}\sim \text{InvGamma}(a,\widehat{\beta}_j)$ is a posterior sample of the asymptotic variance of $\widehat{\tau}^j_{\mathcal{O}}$, and  $\widehat{\beta}_j\sim \text{Gamma}(a_0+r_m^j\alpha,\lambda_0+\mathbf{1}_{r_m^j>0}\widehat{\Upsilon}^{-1}_{m,jj})$ is a posterior sample of the hyperparameter. Then,
    
\begin{align*}
   R_{m+1}(k) &= \sum_{j=1}^Jw_{jj}\widehat{\Upsilon}_{m+1,jj}(k) + \tr\left(\bH_m^\T\bD\bH_m\widehat{\bGamma}_{m}\right) \\
   &- \frac{1}{c}\left[\tr \left\{(\mathbf{I}-\bH_m)^\T\bD\bH_m\widehat{\bGamma}_{m}\right\}-\sum_{j=1}^Jw_{jj}\widehat{\Upsilon}_{m+1,jj}(k)\right]^2.
\end{align*}
Recall that $\widehat{\Gamma}_m$ is the plug-in estimate of the asymptotic variance of observational ATE $\widehat{\btau}_{\mathcal{O},m}$ so that 
 $\widehat{\Gamma}_{m,ij}\overset{p}{\to}0$ for all $1\leq i\leq j\leq J$. Further, since $\widehat{\Upsilon}_{m+1,jj}$ is bounded with probability 1, we can simplify $ R_{m+1}(k)$ as
\[
 R_{m+1}(k) =\left\{\sum_{j=1}^Jw_{jj}\widehat{\Upsilon}_{m+1,jj}(k)\right\}\left\{1-\frac{1}{c} \sum_{j=1}^Jw_{jj}\widehat{\Upsilon}_{m+1,jj}(k)\right\}+ o_p(1).
\]
Thus, minimizing $R_{m+1}(k)$ is asymptotically equivalent to minimizing $\sum_{j=1}^J w_{jj}\widehat{\Upsilon}_{m+1,jj}(k)$ and consequently, for large enough $n$, we have
\begin{align*}
   &P\left(E_m^c\left| \bigcap_{t=n}^{m-1} E_t^c\right.\right) \\
   &= P\left(\left\{j\notin \arg\min_k \left\{\sum_{j=1}^J w_{jj}\widehat{\Upsilon}_{m,jj}(k)\right\}\right\}\left| \bigcap_{t=n}^{m-1}\left\{j\notin \arg\min_k \left\{\sum_{j=1}^J w_{jj}\widehat{\Upsilon}_{t,jj}(k)\right\}  \right.\right\}\right). 
\end{align*}

Since $\widehat{\Upsilon}_{m,jj}(k) = (\overline{L}_{m-1}^j+\mathbf{1}_{k=j}L_{m})^{-1}\widehat{\Upsilon}_{jj}$ , we have
\begin{align*}
    \arg\min_k \left\{\sum_{j=1}^J w_{jj}\widehat{\Upsilon}_{m,jj}(k)\right\} & =   \arg\min_k \left\{\sum_{j=1}^J \frac{w_{jj}\widehat{\Upsilon}_{jj}}{\overline{L}_{m-1}^j+\mathbf{1}_{k=j}L_{m}}\right\} \\
    &=\arg\max_k  \left\{\sum_{j=1}^J \frac{w_{jj}\widehat{\Upsilon}_{jj}}{\overline{L}_{m-1}^j}- \sum_{j=1}^J \frac{w_{jj}\widehat{\Upsilon}_{jj}}{\overline{L}_{m-1}^j+\mathbf{1}_{k=j}L_{m}}\right\}\\
    &= \arg\max_k\left\{w_{kk} \widehat{\Upsilon}_{kk}\left(\frac{1}{\overline{L}^k_{m-1}}-\frac{1}{\overline{L}^k_{m-1}+L_m}\right)\right\}.
\end{align*}

Define $\Delta(m,k):=1/\overline{L}_{m-1}^k - 1/(\overline{L}_{m-1}^{k}+L_m)$. Conditional on the events $\cap_{t=n}^{m-1}E_t^c$ and using the fact that $f(x) = 1/x-1/(x+t)$ is a decreasing function of $x$, we have $\Delta(t,j)=\Delta(n,j)$  for all $n\leq t\leq m-1$ and  $\Delta(m,j')\leq\Delta(m-1,j')\cdots \leq \Delta(n,j')$ for $j'\neq j$. Recall that $w_{jj}>0$ for all $j$ and $\widehat{\Upsilon}_{jj}$ follows a Gamma distribution with mean
\[
\E_{\widehat{\Upsilon}^{-1}_{m,jj}}\frac{a+r_m^j\alpha}{(a_0-1)\left(\lambda_0+\mathbf{1}_{r_m^j>0}\widehat{\Upsilon}^{-1}_{m,jj}\right)}.
\]
We therefore show
\begin{align*}
   P\left(E_m^c\left| \bigcap_{t=n}^{m-1} E_t^c\right.\right)&=P\left(\bigcap_{j'\neq j}\left\{w_{jj}\widehat{\Upsilon}_{jj}\Delta(n,j) < w_{j'j'}\widehat{\Upsilon}_{j'j'}\Delta(m,j')\right\}\right)\\
&\leq P\left(\bigcap_{j'\neq j}\left\{w_{jj}\widehat{\Upsilon}_{jj}\Delta(n,j) < w_{j'j'}\widehat{\Upsilon}_{j'j'}\Delta(m-1,j')\right\}\right)\\
&= P\left(E_{m-1}^c\left| \bigcap_{t=n}^{m-2} E_t^c\right.\right).
\end{align*}

\end{proof}

We will now prove that $E_m$ occurs infinitely often, i.e. $P(\cap_{n=1}^\infty\cup_{m=n}^\infty E_m)=1$. First, note that
\begin{align*}
P\left(\bigcap_{n=1}^\infty\bigcup_{m=n}^\infty E_m\right)=1&\Leftrightarrow P\left\{\left(\bigcap_{n=1}^\infty\bigcup_{m=n}^\infty E_m\right)^c\right\}=0\\
&\Leftrightarrow  P\left(\bigcup_{n=1}^\infty\bigcap_{m=n}^\infty E_m^c\right)=0\\
&\Leftrightarrow \lim_{n\to\infty} P\left(\bigcap_{m=n}^\infty E_m^c\right)=0.
\end{align*}
Since for all $n$ and $N_0>n$, we have
\begin{align*}
    P\left(\bigcap_{m=n}^{N_0} E_m^c\right) &= \left\{\prod_{m=n+1}^{N_0} P\left(E_m^c\left| \bigcap_{t=n}^{m-1} E_t^c\right.\right)\right\} P(E_n^c)\leq P(E_n^c)^{N_0-n+1},
\end{align*}
where the inequality holds because $P\left(E_m^c\left| \bigcap_{t=n}^{m-1} E_t^c\right.\right)\leq P\left(E_{m-1}^c\left| \bigcap_{t=n}^{m-2} E_t^c\right.\right)$ for all $m>n$. Thus, for all $n\in\mathbb{N}$, 
\[
P\left(\bigcap_{m=n}^\infty E_m^c\right) = \lim_{N_0\to\infty} P\left(\bigcap_{m=n}^{N_0} E_m^c\right)=0.
\]
The proof of Lemma 3.3 (i) is therefore complete. The second part of Lemma 3.3 is a straightforward result of law of large numbers and the consistent assumption of the debiased model.

\section{Consistency and asymptotic normality}
\subsection*{Marginal logistic regression model for propensity scores}
Let $\psi(\bx) \in \mathbb{R}^{d}$ be a fixed feature vector, and   
define $\mathrm{expit}(u) = \exp(u)/\left\lbrace 1 + \exp(u) \right\rbrace$.
For any and $\bbeta \in \mathbb{R}^d$ define
$\eta(\bx; \bbeta) \triangleq \mathrm{expit}\left\lbrace \psi(\bx)^{\T}\bbeta\right\rbrace$.
We 
consider a working marginal logistic regression working model for the 
propensity scores in the observational data of the form 
\begin{eqnarray*}
    P(\bA|\bX; \overline{\bbeta}^J) &=& \prod_{j=1}^JP(A^j|\bX; \bbeta^j) = \sum_{j=1}^J \ell(\bX, A^j; \bbeta^j)
\end{eqnarray*}
where 
\begin{equation*}
    \ell(A^j, \bX; \bbeta^j) = \eta(\bX;\bbeta^j)^{A^j}
    \left\lbrace 1 - \eta(\bX; \bbeta^j) \right\rbrace^{1-A^j}
\end{equation*}
is the usual logistic regression model, 
and $\overline{\bbeta}^J = (\bbeta^{1\T}, \ldots, \bbeta^{J\T})^{\T}$ is a vector of unknown
coefficients.    Note that we do not assume this model is correctly specified.  

Let $\ell(\bX, \bA; \overline{\bbeta}^J) = \prod_{j=1}^J\ell(\bX, A^j; \bbeta^j)$.
Given a sample $\left\lbrace (\bX_i, \bA_i, Y_i)\right\rbrace_{i=1}^n$ comprising
$n$ independent copies of $(\bX, \bA, Y)$, we construct and estimator 
$\widehat{\overline{\bbeta}}_n = (\widehat{\bbeta}_n^{1\T},\ldots, \widehat{\bbeta}_n^{J\T})^{\T}$ 
using maximum likelihood, i.e., 
$\widehat{\overline{\bbeta}}_n$ solves 
\begin{equation*}
    \pn \nabla \log\,\ell(\bX, \bA;\overline{\bbeta}^J) = 0,
\end{equation*}
where $\pn$ is the empirical measure.  Define $\overline{\bbeta}^*$ to be the solution  to 
\begin{equation}\label{popnLogisticRegression}
P \nabla\ell(\bX, \bA; \bbeta) = 0.
\end{equation}
We assume (B0) that a solution to (\ref{popnLogisticRegression}) exists and is unique (this assumption is mild as the loss function is strictly convex).   
In addition, we assume 
\begin{itemize}
    \item[(B1)] $P||\nabla\ell(\bX, \bA;\overline{\bbeta})||^2 < \infty$ for all $\overline{\bbeta}$ in a neighborhood
    of $\overline{\bbeta}^*$; 
    \item[(B2)] $\mathbf{H} \triangleq P\nabla^2\ell(\bX, \bA; \overline{\bbeta}^*)$ exists and is strictly positive definite.
\end{itemize}
These assumptions are quite weak.  It follows from standard theory for $M$-estimation 
\citep[e.g., see][]{niemiro1992asymptotics} that 
\begin{eqnarray*}
    \rtn(\widehat{\overline{\bbeta}}_n - \overline{\bbeta}^*) &=& -\rtn(\pn-P)\bH^{-1}\nabla\ell(\bX, \bA; \overline{\bbeta}^*) + o_P(1)\\[3pt]
    &\overset{d}{\to} & \mathcal{N}(0, \bH^{-1}\bGamma \bH^{-1}),
\end{eqnarray*}
where $\bGamma = P\nabla\ell(\bX, \bA; \overline{\bbeta}^*) 
\nabla\ell^{\T}(\bX, \bA;\overline{\bbeta}^*)$.

\subsection*{Asymptotic distribution of $\widehat{\btau}_{\mathcal{O},m}$}
Let $\mathbb{P}_n$ be the empirical distribution on $\{(\bX_i,\bA_i,Y_i)\}_{i=1}^n$. The IPWE estimator of the $j$th treatment effect is

\[
\widehat{\tau}_{\mathcal{O},n}^j \triangleq \frac{\mathbb{P}_n A^j Y / \eta\left(\mathbf{X} ; \widehat{{\boldsymbol{\beta}}}_n^j\right)}{\mathbb{P}_n A^j / \eta\left(\mathbf{X} ; \widehat{\bar{\beta}}_n^j\right)}-\frac{\mathbb{P}_n\left(1-A^j\right) Y /\left\{1-\eta\left(\mathbf{X} ; \widehat{{\boldsymbol{\beta}}}_n^j\right)\right\}}{\mathbb{P}_n\left(1-A^j\right) /\left\{1-\eta\left(\mathbf{X} ; \widehat{{\boldsymbol{\beta}}}_n^j\right)\right\}}
\]
Define the population analog of  $\widehat{\tau}_{\Theta, n}^j$ as 
\[
\tau_{\mathcal{O}}^{j*} \triangleq \frac{P A^j Y / \eta\left(\mathbf{X} ; {\boldsymbol{\beta}}^{j*}\right)}{P A^j / \eta\left(\mathbf{X} ; {\boldsymbol{\beta}}^{j*}\right)}-\frac{P\left(1-A^j\right) Y /\left\{1-\eta\left(\mathbf{X} ; {\boldsymbol{\beta}}^{j*}\right)\right\}}{P\left(1-A^j\right) /\left\{1-\eta\left(\mathbf{X} ; {\boldsymbol{\beta}}^{j*}\right)\right\}}
\]

Define the first and second terms in $\widehat{\tau}^j_{\mathcal{O},n}$ as $\widehat{\tau}^j_{\mathcal{O},n}(1)$ and $\widehat{\tau}^j_{\mathcal{O},n}(0)$, respectively. Similarly, we can define $\tau_{\mathcal{O}}^{j*}(1)$ and $\tau_{\mathcal{O}}^{j*}(0)$ as the corresponding terms in $\tau_{\mathcal{O}}^{j*}.$ The estimators $(\widehat{\tau}^j_{\mathcal{O},n}(1),\widehat{\tau}^j_{\mathcal{O},n}(0),\widehat{\balpha}_1^j,\widehat{\balpha}_0^j,\widehat{{\bbeta}}_n^j)$ can be derived by solving $(\mu_1^j,\mu_0^j,\balpha_1^j,\balpha_0^j,\bbeta^j)$ for


\begin{equation}\label{eq:est}
    \left\{
\begin{array}{lc}
  \mathbb{P}_n A^j Y / \eta\left(\mathbf{X} ; {{\boldsymbol{\beta}}}^j\right) -\mu_1^j\mathbb{P}_n A^j / \eta\left(\mathbf{X} ; {\overline{\beta}}^j\right)  &=0  \\
  \mathbb{P}_n\left(1-A^j\right) Y /\left\{1-\eta\left(\mathbf{X} ; {{\boldsymbol{\beta}}}^j\right)\right\}-\mu_0^j \mathbb{P}_n\left(1-A^j\right) /\left\{1-\eta\left(\mathbf{X} ; {{\boldsymbol{\beta}}}^j\right)\right\}  & =0\\
\mathbb{P}_n \left\{A^j-\eta\left(\bX;\bbeta^j\right)\right\}\bpsi(\bX)&=0,
\end{array}\right.
\end{equation}
where the last estimation equation results from the score function of the propensity score likelihood. Define $\boldsymbol{\gamma}:=\left(\mu_1^j, \mu_0^j, \boldsymbol{\beta}^j\right)$. Let $U_1\left\{\mathbf{A}^j,Y, \mathbf{X} ; \boldsymbol{\gamma}\right\}, U_2\left\{\mathbf{A}^j,Y, \mathbf{X} ; \boldsymbol{\gamma}\right\}$, and $\bU_3\left\{\mathbf{A}^j,Y, \mathbf{X} ; \boldsymbol{\gamma}\right\}$ be the corresponding estimating equations in~\eqref{eq:est}. We will use these estimating equations to derive asymptotic distributions of $\widehat{\gamma}$ in the following sections.

We first show that $\widehat{\tau}^j_{\mathcal{O},n}(1)$ converges in probability to $\tau^{j*}_{\mathcal{O}}(1)$.
 Assume (B3) $\psi(\cdot)<\infty$ and (B4) $\eta(\bX;\bbeta^j)\in(\epsilon,1-\epsilon)$ for $\bbeta_j$ in a neighborhood of $\bbeta^{j*}$ and some $\epsilon>0$.
 
 Note that
\[
\mathbb{P}_n \frac{A^j}{\eta(\bX;\widehat{\overline{\bbeta}})} -\mathbb{P} \frac{A^j}{\eta(\bX;{\bbeta}^{j*})} = \mathbb{P}_n \frac{A^j}{\eta(\bX;\widehat{\overline{\bbeta}})}-\mathbb{P}_n \frac{A^j}{\eta(\bX;{\overline{\bbeta}}^{j*})}+\mathbb{P}_n \frac{A^j}{\eta(\bX;{\overline{\bbeta}}^{j*})} -\mathbb{P} \frac{A^j}{\eta(\bX;{\bbeta}^{j*})}.
\]
Since the class of maps $(\bx,\ba)\to \ba/\eta(\bx;\bbeta^j)$ for $\bbeta^j$ in a neighborhood of $\bbeta^{j*}$ is Glivenko-Cantelli, we have
$
\|\mathbb{P}_n {A^j}/{\eta(\bX;{\overline{\bbeta}}^{j*})} -\mathbb{P} {A^j}/{\eta(\bX;{\bbeta}^{j*})}\|_\infty \to 0 
$
almost surely. In addition, we have
\begin{align*}
\frac{1}{\eta\left(\mathbf{X} ; \widehat{\boldsymbol{\beta}}_n^j\right)}-\frac{1}{\eta\left(\mathbf{X} ; \boldsymbol{\beta}^{j*}\right)} & =-\frac{1}{\eta\left(\mathbf{X} ; \widehat{\boldsymbol{\beta}}_n^j\right) \eta\left(\mathbf{X} ; \boldsymbol{\beta}^{j*}\right)}\left\{\eta\left(\mathbf{X} ; \widehat{\boldsymbol{\beta}}_n^j\right)-\eta\left(\mathbf{X} ; \boldsymbol{\beta}^{j*}\right)\right\} \\
& =-\frac{\nabla \eta\left(\mathbf{X} \mid \boldsymbol{\beta}^{j*}\right)^{\T}\left(\widehat{\boldsymbol{\beta}}_n^j-\boldsymbol{\beta}^{j*}\right)}{\eta\left(\mathbf{X} ; \widehat{\boldsymbol{\beta}}_n^j\right) \eta\left(\mathbf{X} ; \boldsymbol{\beta}^{j*}\right)}+o_P(1 / \sqrt{n}) \\
& =-\frac{\left\{1-\eta\left(\mathbf{X} ; {\boldsymbol{\beta}}^{j*}\right)\right\} \psi(\mathbf{X})^{\T}\left(\widehat{\boldsymbol{\beta}}_n^j-\boldsymbol{\beta}^{j*}\right)}{\eta\left(\mathbf{X} ; \widehat{\boldsymbol{\beta}}_n^j\right)}+o_p(1 / \sqrt{n})
\end{align*}
By assumptions (B4) and (B5) and the consistency of $\widehat{\bbeta}_n^j$, we show that $\mathbb{P}_n {A^j}/{\eta(\bX;\widehat{\overline{\bbeta}})}-{A^j}/{\eta(\bX;{\overline{\bbeta}}^{j*})}$ converges in probability to 0. Following similar arguments, we can show that the nominator $\mathbb{P}_n {A^jY}/{\eta(\bX;\widehat{\overline{\bbeta}})}$ converges in probability to $\mathbb{P} {A^jY}/{\eta(\bX;{\overline{\bbeta}}_0)}$. Thus, by Slustky's theorem and the positivity assumption, we have $\widehat{\tau}^j_{\mathcal{O},n}(1)\overset{p}{\to}\tau^{j*}_{\mathcal{O}}(1).$ Similarly, $\widehat{\tau}^j_{\mathcal{O},n}(0)\overset{p}{\to}\tau^{j*}_{\mathcal{O}}(0).$ Applying Slutsky's theorem again, we thus prove that $\widehat{\tau}^j_{\mathcal{O},n}\overset{p}{\to}\tau^{j*}_{\mathcal{O}}.$

To show the asymptotic noramlity of the proposed estimator, we will start with the asymptotic normality of $\widehat{\tau}^j_{\mathcal{O},n}$.
Define $U^j_{i}:=U_i\{\bA^j,Y,\bX;\bgamma\}$ for $i=1,2,3$ and $\bU^j:=(U^j_{1},U^j_{2},\bU_{3}^{j\T})^\T$. Let $\bU^j_{i0}$ and $\bU_0^j$ be the values of $\bU^j_i$ and $\bU^j$ at $\bgamma^*$, respectively. In addition to assumptions (B1)-(B3) for the asymptotics of $\widehat{\bbeta}^j_n$, we assume

\begin{itemize}
\item[(C1)] $\mathbb{P} (U^j_1,U^j_2,\bU_3^{j\T})^\T$ has a unique solution at $\bgamma^*=(\tau^{j*}_{\mathcal{O}}(1),\tau^{j*}_{\mathcal{O}}(0),{\bbeta}^{j*})$.
    \item[(C2)] $\mathbb{P} \|\bU^j_{0}\|^2<\infty$ for all $j$.
    \item[(C3)] $\bU^j$ is differentiable at $\bgamma^*$ and $\mathbb{P} \left.\partial \bU^j/\partial \bgamma\right|_{\bgamma^*}$ is strictly positive definite for all $j$.
\end{itemize}

By Taylor expansion, we have

$$
\sqrt{n}\left(\begin{array}{c}
\widehat{\tau}_{\mathcal{O}, n}^j(1)-\tau_{\mathcal{O}, 0}^j(1) \\
\widehat{\tau}_{\mathcal{O}, n}^j(0)-\tau_{\mathcal{O}, 0}^j(0) \\
\widehat{\boldsymbol{\beta}}^j_n-\boldsymbol{\beta}_0^j
\end{array}\right)=\left(\begin{array}{ccc}
s_1^j & 0 & \mathbf{b}_1^{j \top} \\
0 & s_2^j & \mathbf{b}_2^{j \top} \\
0 & 0 & \mathbf{H}^j
\end{array}\right)^{-1}\left(\begin{array}{c}
\sqrt{n} U_{10}^j \\
\sqrt{n} U_{20}^j \\
\sqrt{n} U_{30}^j
\end{array}\right)+o_p(1)
$$

where

$$
\begin{aligned}
& s_1^j=\mathbb{P} \partial U_1^j /\left.\partial \mu_1^j\right|_{\boldsymbol{\gamma}_0}=-\mathbb{P}_{\mathbf{X}}\left\{\frac{P\left(A^j=1 \mid \mathbf{X}\right)}{\eta\left(\mathbf{X} ; \boldsymbol{\beta}_0^j\right)}\right\} \\
& s_2^j=\mathbb{P} \partial U_2^j /\left.\partial \mu_0^j\right|_{\boldsymbol{\gamma}_0}=-\mathbb{P}_{\mathbf{X}}\left\{\frac{P\left(A^j=0 \mid \mathbf{X}\right)}{1-\eta\left(\mathbf{X} ; \boldsymbol{\beta}_0^j\right)}\right\} \\
& \mathbf{b}_1^j=\mathbb{P} \partial U_1^j /\left.\partial \boldsymbol{\beta}^j\right|_{\gamma_0}=-\mathbb{P} A^j\left\{Y-\tau_{\mathcal{O}, 0}^j(1)\right\} \exp \left\{-\psi(\mathbf{X})^{\top} \boldsymbol{\beta}_0^j\right\} \psi(\mathbf{X}) \\
& \mathbf{b}_2^j=\mathbb{P} \partial U_2^j /\left.\partial \boldsymbol{\beta}^j\right|_{\boldsymbol{\gamma}_0}=\mathbb{P}\left(1-A^j\right)\left\{Y-\tau_{\mathcal{O}, 0}^j(0)\right\} \exp \left\{\psi(\mathbf{X})^{\top} \boldsymbol{\beta}_0^j\right\} \psi(\mathbf{X}) \\
& \mathbf{H}^j=\mathbb{P} \partial U_3^j /\left.\partial \boldsymbol{\beta}^j\right|_{\boldsymbol{\gamma}_0}=-\mathbb{P}_{\mathbf{X}}\left[\left\{1-\eta\left(\mathbf{X} ; \boldsymbol{\beta}_0^j\right)\right\} \eta\left(\mathbf{X} ; \boldsymbol{\beta}_0^j\right) \psi(\mathbf{X}) \psi(\mathbf{X})^{\top}\right]
\end{aligned}
$$

Here $\mathbf{H}^j$ is the $j$-th block on the diagonal of $\mathbf{H}$, which was defined in the asymptotics of $\widehat{\overline{\boldsymbol{\beta}}}_n=\left(\widehat{\boldsymbol{\beta}}_n^{1 \top}, \cdots, \widehat{\boldsymbol{\beta}}_n^{J \top}\right)^{\top}$. Note that when the propensity score model is correctly specified, we have $s_1^j=s_2^j=1$. By block matrix inverse formula, we can show that

$$
\left(\begin{array}{ccc}
s_1^j & 0 & \mathbf{b}_1^{j \top} \\
0 & s_2^j & \mathbf{b}_2^{j \top} \\
0 & 0 & \mathbf{H}^j
\end{array}\right)^{-1}=\left(\begin{array}{ccc}
s_1^{-j} & 0 & -s_1^{-j} \mathbf{b}_1^{j \top} \mathbf{H}^{-j} \\
0 & s_2^{-j} & -s_2^{-1} \mathbf{b}_2^{j \top} \mathbf{H}^{-j} \\
0 & 0 & \mathbf{H}^{-j}
\end{array}\right)
$$

where $a^{-j}$ denotes as $\left(a^j\right)^{-1}$. Therefore,

\begin{equation*}  \sqrt{n}\left(\widehat{\tau}^j_{\mathcal{O},n}-\tau^{j*}_{\mathcal{O}}\right) \overset{d}{\to}\mathcal{N}\left\{0,(\mathbf{e}_1-\mathbf{e}_2)^\top\E\left( \left.\frac{\partial\bU^j}{\partial\bgamma}\right|_{\bgamma_0} \right)^{-1}\E\bU^j_0\bU_0^{j\T}\E\left( \left.\frac{\partial\bU^j}{\partial\bgamma}\right|_{\bgamma_0} \right)^{-\T}(\mathbf{e}_1-\mathbf{e}_2)\right\},
\end{equation*}
where $\mathbf{e}_j$ is the standard basis vector with $j$-th entry being 1 and 0 otherwise. Specifically, when both propensity score and outcome models are correctly specified, the asymptotic variance of $\widehat{\tau}^j_{\mathcal{O},n}$ reduces to 
\begin{equation*}
    \E U_{10}^jU_{10}^{j\T} + \E U_{20}^jU_{20}^{j\T}.
\end{equation*}
To derive asymptotic distribution of $\sqrt{n}(\widehat{\btau}_{\mathcal{O},n}-\btau_{\mathcal{O}}^*)$, we assume that assumptions $(C1)-(C3)$ hold for $j=1,\cdots,J$. Following the above approach, we can derive $(\widehat{\tau}_{\mathcal{O},n}^1,\cdots,\widehat{\tau}_{\mathcal{O},n}^J)$ by simultaneously solving estimating equations $(\bU_1^\T, \cdots, \bU_J^\T)^\T$, where $\bU^j = (U_1^j, U_2^j, \bU_3^{j\T})$. Thus, the asymptotic distribution of $\widehat{\tau}_{\mathcal{O},n}$ is:
\[
\sqrt{n}(\widehat{\btau}_{\mathcal{O},n}-\btau_{\mathcal{O}}^*)\overset{d}{\to} \mathcal{N}\left(\mathbf{0},\bGamma\right),
\]
where
\[
\bGamma=\E\left( \left.\frac{\partial\bU}{\partial\bgamma}\right|_{\bgamma_0} \right)^{-1}\E\bU_0\bU_0^{\T}\E\left( \left.\frac{\partial\bU}{\partial\bgamma}\right|_{\bgamma_0} \right)^{-\T}\mathbf{E}^\T,
\]
and
\[
\mathbf{E} = \begin{pmatrix}
    \mathbf{e}_1^\top-\mathbf{e}_2^\top & \mathbf{0}^\T &\cdots  &\cdots & \mathbf{0}^\T\\
     \mathbf{0}^\T & \mathbf{e}_1^\top-\mathbf{e}_2^\top  &  \cdots& \cdots &  \mathbf{0}^\T\\
     \vdots& \vdots &\ddots &\ddots &\vdots\\
     \mathbf{0}^\T &\cdots & \cdots & \mathbf{0}^\T & \mathbf{e}_1^\top-\mathbf{e}_2^\top
\end{pmatrix}.
\]

\end{appendices}

\end{document}